\documentclass[11pt]{article}
\usepackage[T1]{fontenc}
\usepackage{lmodern}
\usepackage{geometry} 
\usepackage{amsmath,amssymb,amsthm,amsfonts,dsfont,mathtools}
\usepackage{mathrsfs}
\usepackage[linesnumbered,ruled,vlined]{algorithm2e}
\usepackage{hyperref}
\usepackage[inline]{enumitem}
\usepackage{float}
\usepackage{thmtools}
\usepackage{thm-restate}

\usepackage[dvipsnames]{xcolor}
\usepackage{graphicx}
\setitemize{noitemsep,topsep=4pt,parsep=4pt,partopsep=4pt}

\definecolor{darkgreen}{rgb}{0,0.5,0}
\usepackage{hyperref}
\hypersetup{
    unicode=false,          
    colorlinks=true,        
    linkcolor=RoyalBlue,          
    citecolor=Plum,        
    filecolor=magenta,      
    urlcolor=cyan           
}
\usepackage[capitalize, nameinlink]{cleveref}

\crefname{theorem}{Theorem}{Theorems}
\Crefname{lemma}{Lemma}{Lemmas}
\Crefname{figure}{Figure}{Figures}

\newcommand{\ceil}[1]{\left\lceil #1 \right\rceil}

\newcommand{\poly}{\operatorname{poly}}
\newcommand{\LOCAL}{\mathsf{LOCAL}} 
\newcommand{\CONGEST}{\mathsf{CONGEST}}
\newcommand{\LCL}{\mathsf{LCL}}
\newcommand{\MM}{\mathcal{M}}
\newcommand{\GG}{\mathcal{G}}
\newcommand{\PP}{\mathcal{P}}
\newcommand{\CC}{\mathcal{C}}
\newcommand{\Algo}{\mathcal{A}}
\newcommand{\SSS}{\mathcal{S}}
\newcommand{\LLLL}{\mathsf{L}}
\newcommand{\RRRR}{\mathsf{R}}
\newcommand{\UUUU}{\mathsf{U}}
\newcommand{\DDDD}{\mathsf{D}}
\newcommand{\AAA}{\mathscr{A}}
\newcommand{\BBB}{\mathscr{B}}
\newcommand{\CCC}{\mathscr{C}}

\newcommand{\dist}{\mathsf{dist}}
\newcommand{\colors}{\mathsf{color}}
\newcommand{\upp}{{\sf Up}}
\newcommand{\downn}{{\sf Down}}
\newcommand{\leftt}{{\sf Left}}
\newcommand{\rightt}{{\sf Right}}
\newcommand{\Lovasz}{Lov\'{a}sz}
\newcommand{\LabelIn}{\Sigma_{\operatorname{in}}}
\newcommand{\LabelOut}{\Sigma_{\operatorname{out}}}

\newtheorem{theorem}{Theorem}
\newtheorem{lemma}{Lemma}
\newtheorem{conjecture}{Conjecture}

\newtheorem{definition}{Definition}
\newtheorem{corollary}{Corollary}

\newtheorem{proposition}{Proposition}
\newtheorem{observation}{Observation}

\title{The Distributed Complexity of Locally Checkable Labeling Problems Beyond Paths and Trees}
\author{Yi-Jun Chang\footnote{National University of Singapore. Email: cyijun@nus.edu.sg}}
 
\begin{document}
\date{}
\maketitle
\thispagestyle{empty}
\setcounter{page}{0}

\begin{abstract}
We consider \emph{locally checkable labeling} ($\LCL$) problems in the $\LOCAL$ model of distributed computing. 
Since 2016, there has been a substantial body of work examining the possible complexities of $\LCL$ problems. For example, it has been established that there are no $\LCL$ problems exhibiting deterministic complexities falling between $\omega(\log^\ast n)$ and $o(\log n)$.
This line of inquiry has yielded a wealth of algorithmic techniques and insights that are useful for algorithm designers.

While the complexity landscape of $\LCL$ problems on general graphs, trees, and paths is now well understood, graph classes beyond these three cases remain largely unexplored. Indeed, recent research trends have shifted towards a fine-grained study of special instances within the domains of paths and trees.

In this paper, we generalize the line of research on characterizing the complexity landscape of $\LCL$ problems to a much broader range of graph classes. We propose a  conjecture that characterizes the complexity landscape of $\LCL$ problems for an  \emph{arbitrary} class of graphs that is closed under minors, and we prove a part of the conjecture. 

Some highlights of our findings are as follows.
\begin{itemize}
\item We establish a simple characterization of the minor-closed graph classes sharing the same deterministic complexity landscape as paths, where $O(1)$, $\Theta(\log^\ast n)$, and $\Theta(n)$ are the only possible complexity classes. 

\item It is natural to conjecture that any minor-closed graph class shares the same complexity landscape as trees if and only if the graph class has bounded treewidth and unbounded pathwidth. We prove the ``only if'' part of the conjecture.

\item For the class of graphs with pathwidth at most $k$, we show the existence of $\LCL$ problems with randomized and deterministic complexities $\Theta(n), \Theta(n^{1/2}), \Theta(n^{1/3}), \ldots, \Theta(n^{1/k})$ and the non-existence of  $\LCL$ problems whose deterministic complexity  is between $\omega(\log^\ast n)$ and $o(n^{1/k})$. Consequently, in addition to the well-known complexity landscapes for paths, trees, and general graphs, there are infinitely many different complexity landscapes among minor-closed graph classes.
\end{itemize}
\end{abstract}

\newpage
 
\maketitle

\tableofcontents
\thispagestyle{empty}
\newpage
\section{Introduction}
\setcounter{page}{1}

In the $\LOCAL$ model of distributed computing, introduced by Linial~\cite{Linial92}, a communication network is modeled as an $n$-node graph $G=(V,E)$. In this representation, each node $v \in V$ corresponds to a computer, and each edge $e\in E$ corresponds to a communication link. In each communication round, each node sends a message to each of its neighbors, 
receives a message from each of its neighbors, and then performs some local computation.

The \emph{complexity} of a distributed problem in the $\LOCAL$ model is defined as the smallest number of communication rounds needed to solve the problem, with unlimited local computation power and message sizes. Intuitively, the complexity of a distributed problem is the greatest distance that information needs to traverse within a network to attain a solution, capturing the fundamental concept of \emph{locality} in the field of distributed computing.

There are two variants of the $\LOCAL$ model. In the \emph{deterministic} variant, each node has a unique identifier of $O(\log n)$ bits. In the \emph{randomized} variant, there are no identifiers, and each node has the ability to generate local random bits. Throughout the paper, unless otherwise specified, we assume that the result under consideration applies to both the deterministic and randomized settings.

\subsection{Locally Checkable Labeling}\label{sect:lcl}

A distributed problem on bounded-degree graphs is a \emph{locally checkable labeling} ($\LCL$) if there is some constant $r$ such that the correctness of a solution can be checked locally in $r$ rounds of communication in the $\LOCAL$ model.
The class of $\LCL$  problems encompasses many well-studied problems in distributed computing, including maximal independent set, maximal matching, $(\Delta+1)$ vertex coloring, sinkless orientation, and many variants of these problems.

Formally, an $\LCL$ problem $\PP$ is defined by the following parameters.
\begin{itemize}
    \item An upper bound of the maximum degree $\Delta = O(1)$.
    \item A locality radius $r = O(1)$.
    \item A finite set of input labels $\LabelIn$.
    \item A finite set of output labels  $\LabelOut$.
    \item A set of allowed configurations $\mathcal{C}$. 
\end{itemize}

Each member of $\mathcal{C}$ is a graph $H=(V',E')$ whose maximum degree is at most $\Delta$ with a distinguished \emph{center}  $v \in V'$ such that each node in $H$ is within distance $r$ to $v$ and is assigned an input label from $\LabelIn$ and an output label from $\LabelOut$. The special case of $|\LabelIn| = 1$ corresponds to the case where there is no input label. Since all of $r$, $\Delta$, $\LabelIn$, and $\LabelOut$ are finite,  $\mathcal{C}$ is also finite.

 An \emph{instance} of an $\LCL$ problem $\mathcal{P}$ is a graph $G=(V,E)$  whose maximum degree is at most $\Delta$ where each node is assigned an input label from  $\LabelIn$. A \emph{solution} for  $\mathcal{P}$  on $G$ is a labeling  $\phi_{\text{out}}$ that assigns to each node in $G$ an output label from  $\LabelOut$.  The solution  $\phi_{\text{out}}$ is \emph{correct} if the $r$-radius neighborhood of each node $v \in V$  is isomorphic to an allowed configuration in $\mathcal{C}$ centered at $v$. It is straightforward to generalize the above definition to allow edge orientations and edge labels.

\subsection{The Complexity Landscape of LCL Problems}\label{sect:intro1}
The first systematic study of $\LCL$ problems in the $\LOCAL$ model was done by Naor and Stockmeyer~\cite{NaorS95}. They showed that randomness does not help for $\LCL$ problems whose complexity is $O(1)$, and they also showed that it is \emph{undecidable} to determine whether an $\LCL$ problem can be solved in $O(1)$ rounds.

Since 2016, there has been a substantial body of work examining the possible complexities of $\LCL$ problems~\cite{BrandtEtal16,Balliu18,BalliuBCORS19,balliuBOS18,BBOS20paddedLCL,ChangKP16,ChangP17,chang:LIPIcs:2020:13096,Chang2021automata,FischerG17,grunau2022landscape}. For example, it has been established that there are no $\LCL$ problems exhibiting deterministic complexities falling between $\omega(\log^\ast n)$ and $o(\log n)$~\cite{BrandtEtal16,ChangKP16,PettieS15}.
The complexity landscape of $\LCL$ problems on general graphs, trees, and paths is now well understood.
For trees and paths, complete classifications were known:  The complexity of \emph{any} $\LCL$ problem on trees or paths must belong to one of the following complexity classes.
\begin{description}
\item[Trees:]  $O(1)$, $\Theta(\log^\ast n)$, $\Theta(\log \log n)$, $\Theta(\log n)$, and $\Theta(n^{1/k})$ for each positive integer $k$. 
\item[Paths:] $O(1)$, $\Theta(\log^\ast n)$, and $\Theta(n)$.
\end{description}

All these complexity classes apply to both randomized and deterministic settings, except that the complexity class $\Theta(\log \log n)$ only appears in the randomized setting. Moreover, if an $\LCL$ problem has randomized complexity $\Theta(\log \log n)$ on trees, then its deterministic complexity must be $\Theta(\log n)$ on trees.

\paragraph{Implications.} This line of research is not only interesting from a complexity-theoretic standpoint but has also yielded insights of relevance to algorithm designers.
The derandomization theorem proved in~\cite{ChangKP16} illustrates that the \emph{graph shattering} technique~\cite{BEPS16} employed in many randomized distributed
algorithms gives optimal algorithms to the $\LOCAL$ model.  The \emph{distributed constructive \Lovasz\ local lemma} problem~\cite{ChungPS17} was shown~\cite{ChangP17} to be \emph{complete} for sublogarithmic randomized complexity in a sense similar to the theory of NP-completeness, motivating a series of subsequent research on this problem~\cite{ChangFGUZ19,davies2023improved,FischerG17,ghaffari2018derandomizing,RozhonG20}.

The proof of some of the complexity gaps is \emph{constructive}. For example, the proof that there is no $\LCL$ problem on trees whose complexity is $\omega(\log n)$ and $n^{o(1)}$ given in~\cite{ChangP17} demonstrates an algorithm such that for any given $\LCL$ problem $\PP$ on trees, the algorithm either outputs a description of an $O(\log n)$-round algorithm solving $\PP$ or decides that the complexity of  $\PP$ is $n^{\Omega(1)}$.
Such a result suggests that the design of distributed algorithms could be \emph{automated} in certain settings. 
Indeed, several recent research endeavors in this field have focused on attaining simple characterizations of various complexity classes of $\LCL$ problems that yield 
 efficient algorithms for the automated design of distributed algorithms~\cite{BalliuBCORS19,balliu2020binary,balliu2022locally,balliu2022efficient,brandt2017lcl,chang:LIPIcs:2020:13096,Chang2021automata}. In particular, for $\LCL$ problems on paths without input labels, the task of designing an asymptotically optimal distributed algorithm can be done in polynomial time~\cite{brandt2017lcl,Chang2021automata}. 

 Some of the algorithms for the automated design of distributed algorithms are practical and have been implemented. These algorithms can be used to efficiently discover non-trivial results such as an $O(1)$-round algorithm for maximal independent set on bounded-degree rooted trees~\cite{balliu2022locally}. 

\paragraph{Extensions.}  The study of the complexity landscape of $\LCL$ problems has been extended to other variants of the $\LOCAL$ model: online and dynamic settings~\cite{Akbari2023online}, message size limitation~\cite{balliu2021congest}, volume complexity~\cite{rosenbaum2020volume}, and node-averaged complexity~\cite{balliu2023node}. 
Following the seminal work of Bernshteyn~\cite{bernshteyn2020distributed}, many connections between the complexity classes of $\LCL$ problems in the $\LOCAL$ model and the complexity classes arising from \emph{descriptive combinatorics} have been established~\cite{brandt_et_al:LIPIcs.ITCS.2022.29,grebik2021classification,grebik2021local}.


\subsection{Our Focus: Minor-Closed Graph Classes}

While the complexity landscape of $\LCL$ problems on general graphs, trees, and paths is now well understood, graph classes beyond these three cases remain largely unexplored. Indeed, recent research trends in this field have shifted towards a fine-grained study of special instances within the domains of paths and trees: regular trees~\cite{balliu2022locally,balliu2022efficient,brandt_et_al:LIPIcs.ITCS.2022.29}, rooted trees~\cite{balliu2022locally,balliu2022efficient}, trees with binary input labels~\cite{balliu2020binary}, and paths without input labels~\cite{Chang2021automata}. 

In contrast, a substantial body of work already exists concerning the design and analysis of distributed graph algorithms for various classes of networks beyond paths and trees in the $\LOCAL$ model~\cite{amiri2019distributed,bonamy2021tight,chechik2019optimal,Czygrinow2007cocoon,CZYGRINOW2006JDA,Czygrinow2006ESA,czygrinow2006,czygrinow2008fast,czygrinow2014distributed,czygrinow2020distributed,lenzen2013distributed,wawrzyniak2014strengthened}, so there currently exists a considerable gap between the complexity-theoretic and algorithmic understanding of locality in distributed computing.

To address this issue, let us consider the following generic question:
Can we characterize the set of possible complexity classes of $\LCL$ problems for any given graph class $\GG$?
As any set of graphs is a graph class, it is possible to construct artificial graph classes to realize various strange complexity landscapes.
To obtain meaningful interesting results, we must restrict our attention to some \emph{natural} graph classes.

\paragraph{Minor-closed graph classes.}
In this work, we focus on characterizing the possible complexity classes of $\LCL$ problems on any given \emph{minor-closed} graph class.
The {minor-closed} graph classes are among the most prominent types of sparse graphs, covering many natural sparse graph classes, such as forests, cacti, planar graphs, bounded-genus graphs, and bounded-treewidth graphs.

A graph $H$ is a \emph{minor} of $G$ if $H$ can be obtained from $G$ by removing nodes, removing edges, and contracting edges. Alternatively, $H$ is a minor of $G$ if there exist a partition of $V(G)$ into $k=|V(H)|$ disjoint connected clusters $\CC = \{V_1, V_2, \ldots, V_k\}$ and a bijection between $\CC = \{V_1, V_2, \ldots, V_k\}$ and $V(H)$  such that for each edge $e$ in $H$, the two clusters in  $\CC$ corresponding to the two endpoints of $e$ are adjacent in $G$. 

Any set $\GG$ of graphs is called a \emph{graph class}. 
 A graph class $\GG$ is said to be \emph{minor-closed} if $G \in \GG$ implies that all minors $H$ of $G$ also belong to $\GG$. Alternatively,  a graph class $\GG$ is minor-closed if it is closed under removing nodes, removing edges, and contracting edges.
 
  A cornerstone result in structural graph theory is the \emph{graph minor theorem} of Robertson and Seymour~\cite{ROBERTSON2004325}, which implies that for \emph{any}  minor-closed graph class $\GG$, there is a \emph{finite} list of graphs $H_1, H_2, \ldots, H_k$ such that $G \in \GG$ if and only if $G$ does not contain any of $H_1, H_2, \ldots, H_k$  as a minor. Thus, any minor-closed graph class has a finite description in terms of a list of forbidden minors.
  The ideas developed in the proof of the graph minor theorem hold significance not only for mathematicians but also prove highly valuable in algorithm design and analysis. This has given rise to a thriving research field known as \emph{algorithmic graph minor theory}~\cite{demaine2005algorithmic,demaine2005subexponential}. 
 
 \subsection{Our Contribution}
 The main contribution of this work is the formulation of a  conjecture that characterizes the complexity landscape of $\LCL$ problems for an  \emph{arbitrary} class of graphs that is closed under minors. We present the conjecture in \cref{sect:intro2}. To substantiate the conjecture, we provide a collection of results in \cref{sect:results}, which collectively serve as a partial validation of the proposed conjecture. For the sake of presentation, detailed technical proofs for the assertions made in \cref{sect:intro2,sect:results} are left to \cref{sect:bounded-growth,sect:complexity-landscape,sect:decide,sect-dense-region}. We conclude the paper with a discussion of potential future directions in \cref{sect:conclusions}.

\subsection{Additional Related Work}
 The $\CONGEST$ model of distributed computing is a variant of the $\LOCAL$ model with an $O(\log n)$-bit message size constraint. 
There has been substantial research dedicated to utilizing the structural graph properties of minor-closed graph classes in the design of efficient algorithms in the $\CONGEST$ model: tree decomposition and its applications~\cite{IzumiSPAA22}, low-congestion shortcut and its applications~\cite{ghaffari2021low,haeupler2016low,haeupler2016near,haeupler2018minor,haeuplerLi2018disc,haeupler2018round}, planarity testing and its applications~\cite{ghaffari2016distributed,ghaffari2017near,levi2021property}, expander decomposition and its applications~\cite{chang2023efficient,chang2022narrowing}.

 In \emph{local certification}, labels are assigned to nodes in a network to certify some property of the network. The certification is local in that the checking process is done by an $O(1)$-round $\LOCAL$ algorithm. Researchers have developed local certification algorithms tailored to various minor-closed graph classes~\cite{bousquet2022local,feuilloley2021compact,feuilloley2022can,fraigniaud2022meta,feuilloley2023local,Naor2020soda}. Akin to the study of $\LCL$ problems in the $\LOCAL$ model, some \emph{algorithmic meta-theorems} that apply to a wide range of graph properties have been established for local certification~\cite{feuilloley2022can,fraigniaud2022meta}.

\section{Our Conjecture} \label{sect:intro2}

We first present some terminology that is needed to state our conjecture.

\begin{definition}\label{def-trees}
 For any two positive integers $k$ and $s$, the rooted tree $T_{k,s}$ is defined as follows.
\begin{itemize}
    \item For $k = 1$,  $T_{1,s}$ is an $s$-node path $(v_1, v_2, \ldots, v_s)$, where $v_1$ is designated as the root.
    \item For each $k > 1$, $T_{k,s}$ is constructed as follows. Start from an $s$-node path $(v_1, v_2, \ldots, v_s)$  and $s$ copies $T_1, T_2, \ldots, T_s$ of  $T_{k-1, s}$. For each $j \in [s]$, add an edge connecting $v_j$ and the root of  $T_j$. Designate $v_1$ as the root.
\end{itemize}
\end{definition}

Intuitively, the rooted tree $T_{k,s}$  in \cref{def-trees} can be seen as a $k$-level hierarchical combination of $s$-node paths. 
The significance of $T_{k,s}$ to the complexity landscape of $\LCL$ problems lies in its role as a hard instance. Specifically, both $T_{k,s}$ and its variants have been identified as hard instances for $\LCL$ problems in the complexity class $\Theta(n^{1/k})$. These trees have been employed as lower-bound graphs in \emph{all} existing $\Omega(n^{1/k})$ lower-bound proofs for such $\LCL$ problems~\cite{balliu2022efficient,balliu2022locally,ChangP17,chang:LIPIcs:2020:13096}. See \cref{fig:graph} for an illustration of the construction of $T_{3,3}$ from three copies $T_1$, $T_2$, and $T_3$ of $T_{2,3}$ and a $3$-node path $(v_1, v_2, v_3)$, where the roots are drawn in black.
 
 \begin{definition}\label{def-pw}
 For each non-negative integer $k$, define $\AAA_k$ as the set of all minor-closed graph classes $\GG$ meeting the following two conditions.
\begin{description}
    \item[(C1)] If $k \geq 1$, then $T_{k, s} \in \GG$ for all positive integers $s$.
    \item[(C2)] There exists a positive integer $s$ such that $T_{k+1, s} \notin \GG$.
\end{description}
 \end{definition}

\begin{figure}
	\centering
	\includegraphics[width=\textwidth]{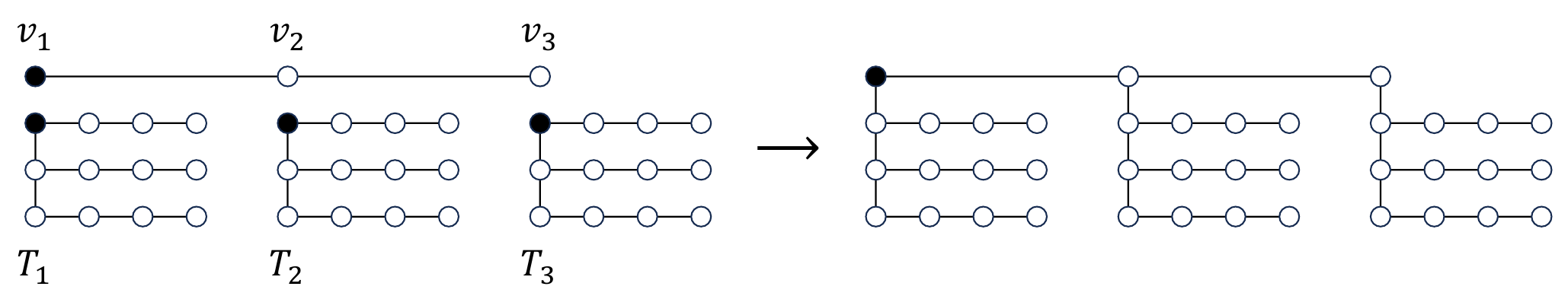}
	\caption{The construction of $T_{3,3}$ in \cref{def-pw}.} 
	\label{fig:graph}
\end{figure}
 
 For each positive integer $k$, $\AAA_k$ is the set of all minor-closed graph classes $\GG$ that includes all of $T_{k,s}$ and excludes some of $T_{k+1,s}$. For the special case of $k=0$, (C1) is vacuously true, so $\AAA_0$ is the set of all minor-closed graph classes $\GG$ that excludes some of $T_{1,s}$.

 

\paragraph{Pathwidth.}
A \emph{path decomposition} of a graph $G=(V,E)$ is a sequence $(X_1, X_2, \ldots, X_k)$ of subsets of $V$ meeting the following conditions.
\begin{description}
    \item[(P1)] $X_1 \cup X_2 \cup \cdots \cup X_k = V$.
    \item[(P2)] For each node $v \in V$, there are two indices $i$ and $j$ such that $v \in X_{l}$ if and only if $i \leq l \leq j$.
    \item[(P3)] For each edge $e=\{u,v\} \in E$, there is an index $i$ such that $\{u,v\} \subseteq X_i$. 
\end{description}

The \emph{width} of a path decomposition $(X_1, X_2, \ldots, X_k)$  is $\max_{1 \leq i \leq k} |X_i| - 1$. The \emph{pathwidth} of a graph is the minimum width over all path decompositions of the graph. Intuitively, the pathwidth of a graph measures how similar it is to a path. A graph class $\GG$ is said to have \emph{bounded pathwidth} if there is a finite number $C$ such that each $G \in \GG$ has pathwidth at most $C$.

It is well-known that for any integer $k$, the set of all graphs with pathwidth at most $k$ is a minor-closed graph class. We emphasize that it is, however, possible that a bounded-pathwidth graph class is not minor-closed. One such example is the set of all graphs with pathwidth at most $k$ and containing at least two nodes.
 
Throughout the paper, we write $\AAA$ to denote the set of all minor-closed graph classes whose pathwidth is bounded. Formally, $\GG \in \AAA$ if $\GG$ is minor-closed and there is a finite number $C$ such that the pathwidth of each $G \in \GG$ is at most $C$.

For the sake of presentation, the proof of the following two statements is deferred to \cref{sect-bounded-pathwidth}.

\begin{restatable}{proposition}{thmpwa}
\label{lem-pathwidth-partition}
$\AAA = \bigcup_{0 \leq i < \infty} \AAA_i$ is a partition of $\AAA$ into disjoint sets.
\end{restatable} 

\cref{lem-pathwidth-partition} shows that $\bigcup_{0 \leq i < \infty} \AAA_i$ is a classification of all  bounded-pathwidth  minor-closed graph classes.

\begin{restatable}{proposition}{thmpwb}
\label{thm-pathwidth-k}
For every integer $k\geq 1$, the class of all graphs with pathwidth at most $k$ is in $\AAA_k$.
\end{restatable}

 
We emphasize that the converse of \cref{thm-pathwidth-k} does not hold in the following sense: $\GG \in \AAA_k$ does not imply that each $G \in \GG$ has pathwidth at most $k$. For example, if we let $\GG$ be the set of all graphs with at most $k+1$ nodes, then $\GG \in \AAA_0$ because it does not contain the $(k+2)$-node path $T_{1, k+2}$, but $\GG$ contains the $(k+1)$-node complete graph, whose pathwidth is $k$.

\paragraph{The conjecture.}
  For any graph class $\GG$ and a positive integer $\Delta$, we write $\GG^\Delta$ to denote the set of graphs in $\GG$ whose maximum degree is at most $\Delta$.
   In the subsequent discussion, we informally say that a range of complexities is  \emph{dense} to indicate that many different complexity functions in this range can be realized by $\LCL$ problems in a way reminiscent of the famous \emph{time hierarchy theorem} for Turing machines.  We are now ready to state our conjecture that characterizes the complexity landscape of $\LCL$ problems for an  \emph{arbitrary} minor-closed graph class.

\begin{conjecture}\label{conjecture-main}
Let $\mathcal{G}$ be a minor-closed graph class that is not the class of all graphs, and let $\Delta\geq 3$ be an integer. The complexity landscape of $\LCL$ problems on $\GG^\Delta$ is characterized as follows.
\begin{itemize}
    \item If $\GG \in \AAA_0$, then $O(1)$ is the only possible complexity class. 
    \item If $\GG \in \AAA_k$ for some $k > 0$, then the  possible complexity classes are exactly \[O(1), \Theta(\log^\ast n), \Theta\left(n^{1/k}\right), \Theta\left(n^{1/(k-1)}\right), \ldots, \Theta(n).\]
    \item If $\GG$ has unbounded pathwidth and bounded treewidth, then  $\GG$ shares the same complexity landscape as trees. In other words, the possible complexity classes are exactly \[O(1), \Theta(\log^\ast n), \Theta(\log \log n), \Theta(\log n), \ \text{and} \ \Theta\left(n^{1/k}\right) \  \text{for all positive integers $k$}.\]
    \item If $\GG$ has unbounded treewidth, then the  possible complexity classes are exactly
    \[O(1), \Theta(\log^\ast n), \Theta(\log \log n), \ \text{and a dense region} \ \left[\Theta(\log n), \Theta(n)\right].\]
\end{itemize}
\end{conjecture}

Similar to the case of trees discussed earlier, all complexity classes in \cref{conjecture-main} apply to both randomized and deterministic settings, except that the complexity $\Theta(\log \log n)$ only exists in the randomized setting.

\section{Supporting Evidence} \label{sect:results}

In this section, we show a collection of
results that prove a part of \cref{conjecture-main}.
We say that an $\LCL$ problem $\PP$ is \emph{solvable} in a graph class $\GG$ if all graphs in $\GG$ admit a correct solution for $\PP$.
Let us first consider $\AAA_0$.
Intuitively, $\GG \in \AAA_0$ means that there is some length bound $\ell$ such that all graphs in $\GG$ do not contain a minor isomorphic to the $\ell$-length path. Therefore, all members in $\GG$ have bounded diameter, so as long as the considered $\LCL$ problem $\PP$ is solvable, $\PP$ can be solved in $O(1)$ rounds in the $\LOCAL$ model by a brute-force information gathering.

\begin{restatable}{theorem}{thmcomplexitya}
\label{thm-path-LCL-1}
Let  $\Delta\geq 3$  be an integer, and let $\GG \in \AAA_0$. All $\LCL$ problems that are solvable in $\GG^\Delta$ can be solved in $O(1)$ rounds in $\GG^\Delta$.
\end{restatable} 

Our classification $\AAA = \bigcup_{0 \leq i < \infty} \AAA_i$ makes sense from two different points of view. The first viewpoint is to consider the special role of the trees $T_{k,s}$ as hard instances in the study of the complexity of $\LCL$ problems.

Recall from an earlier discussion that \emph{all} existing $\Omega(n^{1/k})$ lower-bound proofs for $\LCL$ in the complexity class $\Theta(n^{1/k})$ in trees are based on the tree $T_{k,s}$ or its variants. In view of the definition of $\AAA_k$, for any $\GG \in \AAA_k$. we expect that the complexity class $\Theta(n^{1/s})$ exists for all $s\in\{1, 2, \ldots,k\}$.

\begin{restatable}{theorem}{thmcomplexityb}
\label{thm-path-LCL-2}
Let $k\geq 1$, $\Delta\geq 3$, and $s\in\{1, 2, \ldots,k\}$ be integers, and let $\GG \in \AAA_k$. 
There is an $\LCL$ problem whose complexity in $\GG^\Delta$ is $\Theta(n^{1/s})$.
\end{restatable} 

Similarly, for any $\GG \in \AAA_k$. we expect that the complexity class $\Theta(n^{1/s})$ does not exist for all $s\in\{k+1, k+2, \ldots\}$. To establish this claim, we consider a different perspective.

In \cref{sect:bounded-growth}, we prove an alternative characterization of the set of graph classes $\AAA_k$ in terms of the growth rate of the size of the $d$-radius neighborhood. In particular, we show that $\AAA_k$ is precisely the set of all minor-closed graph classes $\GG$ such that $k$ is the smallest number such that the size of the $d$-radius neighborhood of any node in any bounded-degree graph in $\GG$ is  $O(d^k)$. 

The growth rate of the size of the $d$-radius neighborhood is relevant to the complexity landscape of $\LCL$ problems in that the growth rate affects the complexity gaps resulting from existing approaches. In particular, the alternative characterization of $\AAA_k$, combined with the existing proofs~\cite{ChangKP16,ChangP17,NaorS95} to establish complexity gaps for $\LCL$ problems in general graphs, yields the following two results.

\begin{restatable}{theorem}{thmcomplexityc}
\label{thm-path-LCL-3}
Let $k\geq 1$ and $\Delta\geq 3$ be integers, and let $\GG \in \AAA_k$. There is no $\LCL$ problem whose deterministic complexity in $\GG^\Delta$ is between $\omega(\log^\ast n)$ and $o(n^{1/k})$.
\end{restatable}

\begin{restatable}{theorem}{thmcomplexityd}
\label{thm-path-LCL-4}
Let $k\geq 1$ and $\Delta\geq 3$ be integers, and let $\GG \in \AAA_k$.
There is no $\LCL$ problem whose  complexity in  $\GG^\Delta$ is between $\omega(1)$ and $o\left((\log^\ast n)^{1/k}\right)$.
\end{restatable} 

The proof of \cref{thm-path-LCL-1,thm-path-LCL-2,thm-path-LCL-3,thm-path-LCL-4} are deferred to  \cref{sect:complexity-landscape}.

\paragraph{An infinitude of complexity landscapes.}
\cref{thm-path-LCL-2,thm-path-LCL-3,thm-pathwidth-k} imply that, in addition to the well-known complexity landscapes for paths, trees, and general graphs, there are infinitely many different complexity landscapes among minor-closed graph classes. 

For every $\Delta\geq 3$ and $k\geq 1$, let us consider the class of graphs with maximum degree $\Delta$ and pathwidth at most $k$. For this class of graphs, there exist $\LCL$ problems with randomized and deterministic complexities $\Theta(n), \Theta(n^{1/2}), \Theta(n^{1/3}), \ldots, \Theta(n^{1/k})$, and there does not exist an $\LCL$ problem whose deterministic complexity is between $\omega(\log^\ast n)$ and $o(n^{1/k})$. These results already guarantee that the complexity landscapes necessarily vary for different values of $k$.

\paragraph{Algorithmic implications.} 
\cref{thm-pathwidth-k,thm-path-LCL-3} imply that  any $n^{o(1)}$-round deterministic distributed algorithm $\Algo$ for any $\LCL$ problem on bounded-pathwidth graphs can be automatically turned into an $O(\log^\ast n)$-round deterministic algorithm $\Algo'$ solving  the same problem on bounded-pathwidth graphs. This allows us to \emph{automatically speed up} existing algorithms significantly on bounded-pathwidth graphs.

\begin{corollary}\label{cor-algo}
The following problems can be solved in $O(\log^\ast n)$ rounds deterministically on graphs of bounded pathwidth and bounded degree.
\begin{itemize}
    \item Constructive \Lovasz\ local lemma with the condition $epd \leq 1-\delta$, for any constant $\delta > 0$.
    \item $\Delta$ vertex coloring.
    \item $(\Delta+1)$ edge coloring.
\end{itemize}
\end{corollary}
\begin{proof}
Given the discussion above, we just need to check that these $\LCL$ problems can be solved in $n^{o(1)}$ rounds deterministically on graphs with maximum degree $\Delta = O(1)$. Indeed, constructive \Lovasz\ local lemma can be solved in polylogarithmic rounds on general graphs~\cite{RozhonG20}, $\Delta$ vertex coloring can be solved in $O(\log^2 n)$ rounds on bounded-degree graphs~\cite{ghaffari2018improved}, and $(\Delta+1)$ edge coloring can be solved in $\poly(\Delta, \log n)$ rounds on general graphs~\cite{christiansen2023power}.
\end{proof}

It seems to be a highly nontrivial task to explicitly construct the $O(\log^\ast n)$-round algorithms whose \emph{existence} is guaranteed in \cref{cor-algo}, as the proof of \cref{thm-path-LCL-3} is non-constructive in the sense that it does not offer an algorithm that decides between the two cases $O(\log^\ast n)$ and $\Omega(n^{1/k})$.  

\subsection{Path-Like Graph Classes}
\cref{thm-path-LCL-2,thm-path-LCL-3,thm-path-LCL-4}  allow us to completely characterize the minor-closed graph classes whose deterministic complexity landscape is identical to that of paths: $O(1)$, $\Theta(\log^\ast n)$, and $\Theta(n)$.

\begin{corollary}\label{thm-path-like}
For any minor-closed graph class $\GG$, the possible deterministic complexity classes for $\LCL$ problems on bounded-degree graphs in $\GG$ are exactly $O(1)$, $\Theta(\log^\ast n)$, and $\Theta(n)$ if and only if $\GG \in \AAA_1$.
\end{corollary}
\begin{proof}
Suppose $\GG \in \AAA_1$. Then \cref{thm-path-LCL-3,thm-path-LCL-4} show that $O(1)$, $\Theta(\log^\ast n)$, and $\Theta(n)$ are the only possible deterministic complexity classes for $\LCL$ problems in  bounded-degree graphs in $\GG$. Indeed, it is well-known~\cite{Chang2021automata} that there exist $\LCL$ problems with deterministic complexities $O(1)$, $\Theta(\log^\ast n)$, and $\Theta(n)$ even if we restrict ourselves to path graphs. Since $\GG \in \AAA_1$,  $\GG$ contains all path graphs, so $O(1)$, $\Theta(\log^\ast n)$, and $\Theta(n)$ are exactly the possible deterministic complexity classes for $\LCL$ problems on  bounded-degree graphs in $\GG$. 

Suppose $\GG \notin \AAA_1$. If $\GG \in \AAA_k$ for some $k \neq 1$, then \cref{thm-path-LCL-1,thm-path-LCL-2} show that the set of possible deterministic complexity classes  for $\LCL$ problems on bounded-degree graphs in $\GG$ cannot be $O(1)$, $\Theta(\log^\ast n)$, and $\Theta(n)$. If $\GG \notin \AAA_k$ for all $k$, then $\GG$ has unbounded pathwidth.
The well-known excluding forest theorem~\cite{ROBERTSON1983forest} implies that any
 minor-closed graph class $\GG$ with unbounded pathwidth must contain all trees. Hence the complexity landscape for $\GG^\Delta$ includes all the complexity classes for trees, such as $\Theta(\log n)$, see~\cite{ChangP17}.
\end{proof}

\paragraph{The complexity of the characterization.}
The graph minor theorem~\cite{ROBERTSON2004325} of Robertson and Seymour implies that if $\GG$ is a minor-closed graph class, then there is a finite list of graphs $H_1, H_2, \ldots, H_k$ such that $G \in \GG$ if and only if $G$ does not contain any of $H_1, H_2, \ldots, H_k$  as a minor. Therefore, any minor-closed graph class $\GG$ admits a \emph{finite} representation by listing its finite list of forbidden minors $\{H_1, H_2, \ldots, H_k\}$, so it makes sense to consider computational problems where the input is an arbitrary minor-closed graph class.

A common method to demonstrate the simplicity of a characterization is by establishing its polynomial-time computability.
Given any minor-closed graph class $\GG$, represented by a list of forbidden minors $H_1, H_2, \ldots, H_k$, is there an efficient algorithm deciding whether $\GG$ is path-like in the sense of \cref{thm-path-like}? We show an affirmative answer to this question. In fact, we prove a more general result which shows that for any fixed index $i$, whether $\GG \in \AAA_i$ is decidable in time polynomial in the size of the representation of $\GG$, which is a finite list of forbidden minors $\{H_1, H_2, \ldots, H_k\}$. 

\begin{restatable}{proposition}{thmdecide}
\label{lem-Ak-decide}
For any fixed index $i$, there is a polynomial-time algorithm that, given a list of graphs $H_1, H_2, \ldots, H_k$, decides whether the class of $\{H_1, H_2, \ldots, H_k\}$-minor-free graphs $\GG$  is in $\AAA_i$.
\end{restatable}

The proof of \cref{lem-Ak-decide} is in \cref{sect:decide}. We remark that although the algorithm of \cref{lem-Ak-decide} finishes in polynomial time, the algorithm is unlikely to be practical in that the list of forbidden minors for the considered minor-closed graph class $\GG$ is often not known.

\subsection{Tree-Like Graph Classes}  It is natural to conjecture that any minor-closed graph class shares the same complexity landscape as trees if and only if the graph class has bounded treewidth and unbounded pathwidth. We prove the ``only if'' part of the conjecture.

\begin{corollary}\label{cor-onlyif}
A minor-closed graph class $\GG$ shares the same complexity landscape as trees only if $\GG$ has bounded treewidth and unbounded pathwidth.
\end{corollary}
\begin{proof}
\cref{thm-path-LCL-3} implies that if a minor-closed graph class $\GG$ has bounded pathwidth, then the complexity class $\Theta(\log n)$ disappears. Therefore, for $\GG$ to have the same complexity landscape as that of trees, $\GG$ must have unbounded pathwidth.

Now suppose $\GG$  is a minor-closed graph class with unbounded treewidth. Then the well-known excluding grid theorem~\cite{ROBERTSON1986excludingPlanar} implies that $\GG$ contains all planar graphs. 

It was shown in~\cite{balliuBOS18} that for any rational number $\frac{1}{2} \leq c < 1$, there exists an $\LCL$ problem $\PP$ that is solvable in  $O(n^c)$ rounds for all graphs and requires $\Omega(n^c)$ rounds to solve for planar graphs. As  $\GG$ contains all planar graphs, this implies that the complexity of $\PP$ is $\Theta(n^{c})$ in $\GG$, This shows that $[\Theta(\sqrt{n}), \Theta(n)]$ is a \emph{dense region} in the complexity landscape for $\GG$. Hence the complexity landscape for $\GG$ is different from that of trees.
\end{proof}

\paragraph{The range of the dense region.}
In the above proof, we see that the dense region $[\Theta(\sqrt{n}), \Theta(n)]$ exists in the complexity landscape for any minor-closed graph class $\GG$ that has unbounded treewidth. In \cref{sect-dense-region}, we generalize this result to extend the range of the dense region to cover the entire interval $[\Theta(\log n), \Theta(n)]$, which is the \emph{widest possible} due to the $\omega(\log^\ast n)$ -- $o(\log n)$ gap shown in~\cite{ChangKP16}. Specifically, we show that there are $\LCL$ problems with the following complexities.

\begin{restatable}{theorem}{thmdense}\label{thm-time-hierarchy-main}
For any minor-closed graph class  $\GG$  that has unbounded treewidth, there are $\LCL$ problems on bounded-degree graphs in $\GG$ with the following complexities.
\begin{itemize}
    \item $\Theta(n^{c})$, for each rational number $c$ such that $0 \leq c \leq 1$.
    \item $\Theta(\log^{c} n)$, for each rational number $c$ such that $c \geq 1$.    
\end{itemize}
\end{restatable}

\paragraph{Existing approaches.}
We briefly explain the construction of the $\LCL$ problem $\PP$ in~\cite{balliuBOS18} which is used in the above proof of \cref{cor-onlyif}.
The $\LCL$ problem uses locally checkable constraints to force the underlying network to encode an execution of a linear-space-bounded Turing machine $\MM$ in a two-dimensional grid. Suppose the time complexity of  $\MM$ on the input string  $0^s$ is $T(s)$. Running a simulation of $\MM$ on  $0^s$  requires $n = s \cdot T(s)$
 nodes and $t = T(s)$ rounds. For any $\frac{1}{2} \leq c < 1$, the round complexity function $t = \Theta(n^{c})$ can be realized with $T(s) = s^{c/(1-c)}$.
Since  $T(s) = \Omega(s)$, $[\Theta(\sqrt{n}), \Theta(n)]$ is the largest possible dense region resulting from this approach.

It was shown in~\cite{balliuBOS18} that $\LCL$ problems with complexity $\Theta(n^c)$ exist by extending this construction to higher dimensional grids. This extension is not applicable for proving \cref{thm-time-hierarchy-main}, due to the following reason. For any graph $H$ there exists a number $d$ such that $H$ is a minor of a sufficiently large $d$-dimensional grid. If a minor-closed graph class $\GG$ contains arbitrarily large $d$-dimensional grids for all $d$, then $\GG$ must be the set of all graphs.

With a completely different approach, in another previous work~\cite{Balliu18}, the two dense regions $[\Theta(\log \log^\ast n), \Theta(\log^\ast n)]$ and $[\Theta(\log n), \Theta(n)]$ were shown for $\LCL$ problems on general graphs. Due to a similar reason, their construction of $\LCL$ problems is also not applicable for proving \cref{thm-time-hierarchy-main}.

The proof in~\cite{Balliu18} relies on the following graph structure. Start with an $a \times b$ grid graph whose dimensions $a$ and $b$ can be arbitrarily large. Let $u_{i,j}$  denote the node on the $i$th row and the $j$th column.
For each row $i$, add an edge between $u_{i,j_1}$ and $u_{i,j_2}$ if $j_2 - j_1 = i$. This graph contains a $k$-clique as a minor given that $k \leq \min\{a,b\}$. To see this, contract the $j$th column into a node $v_j$ for each $1 \leq j \leq k$. Then $\{v_1, v_2, \ldots, v_k\}$ forms a clique. Thus, for the results in~\cite{Balliu18} to apply to a minor-closed graph class $\GG$, $\GG$ must contain the $k$-clique for all $k$. Since any graph is a minor of a sufficiently large clique, it follows that $\GG$ is necessarily the set of all graphs.

\paragraph{New ideas.}
To establish the dense region  $[\Theta(\log n), \Theta(n)]$, in \cref{sect-dense-region} we modify the construction of~\cite{balliuBOS18} by attaching a root-to-leaf path of a complete tree to one side of the grid used in the Turing machine simulation. This enables us to increase the number of nodes to be \emph{exponential} in the round complexity of  Turing machine simulation, allowing us to realize any reasonable $\LCL$ complexity in the region $[\Theta(\log n), \Theta(n)]$ and to prove \cref{thm-time-hierarchy-main}.

\subsection{Summary}
 
For convenience, in the subsequent discussion, let $\BBB$ denote the set of all minor-closed graph classes $\GG$ that has unbounded pathwidth and bounded treewidth, and let $\CCC$ denote the set of all minor-closed graph classes $\GG$ that has unbounded treewidth and is not the set of all graphs.

\paragraph{The state of the art.}
We summarize the old and new results on the complexity of $\LCL$ problems as follows. For simplicity, here we only consider the deterministic $\LOCAL$ model.
  \begin{align*}
   \AAA_0 \hspace{0.5cm}  & O\left(1\right) \\
  \AAA_k \hspace{0.5cm}  & O\left(1\right) 
\frac{\times}{\hspace{0.6cm}} 
\overset{?}{\Theta\left(\left(\log^\ast n\right)^{\frac{1}{k}}\right)}
\frac{?}{\hspace{0.6cm}} 
\Theta\left(\log^\ast n\right) 
\frac{\times}{\hspace{0.6cm}} 
\Theta\left(n^{\frac{1}{k}}\right) 
\frac{?}{\hspace{0.6cm}} 
\Theta\left(n^{\frac{1}{k-1}}\right) 
\frac{ ?}{\hspace{0.25cm}} 
\cdots
\frac{ ?}{\hspace{0.25cm}} 
\Theta\left({n^{\frac{1}{2}}}\right) 
\frac{?}{\hspace{0.6cm}} 
\Theta\left(n\right) \\
  \BBB \hspace{0.5cm}  &  O\left(1\right) 
\frac{\times}{\hspace{0.6cm}} 
\overset{?}{\Theta\left(\log \log^\ast n\right)}
\frac{?}{\hspace{0.6cm}} 
\Theta\left(\log^\ast n\right) 
\frac{\times}{\hspace{0.6cm}} 
\Theta\left(\log  n\right) 
\frac{ ?}{\hspace{0.25cm}} 
\cdots
\frac{ ?}{\hspace{0.25cm}} 
\Theta\left(n^{\frac{1}{3}}\right) \frac{?}{\hspace{0.6cm}} 
\Theta\left({n^{\frac{1}{2}}}\right)
\frac{?}{\hspace{0.6cm}} 
\Theta\left(n\right) \\
  \CCC \hspace{0.5cm}  & O\left(1\right) 
\frac{\times}{\hspace{0.6cm}} 
\overset{?}{\Theta\left(\log \log^\ast n\right)}
\frac{?}{\hspace{0.6cm}} 
\Theta\left(\log^\ast n\right) 
\frac{\times}{\hspace{0.6cm}} 
\Theta\left(\log  n\right) 
\frac{\text{dense}}{\hspace{1cm}} 
\Theta\left(n\right) 
 \end{align*}
 The $\omega(\log^\ast n)$ -- $o(\log n)$ and $\omega(1)$ -- $o(\log \log^\ast n)$ gaps for $\BBB$ and $\CCC$ are due to~\cite{ChangKP16,ChangP17}. The existence of the complexity class $\Theta(n^{1/k})$ for all positive integers $k$ for $\BBB$ is due to~\cite{ChangP17}. The existence of the complexity class $\Theta(\log^\ast n)$ for $\AAA_1, \AAA_2, \ldots$, $\BBB$, and $\CCC$ is due to the well-known fact that the complexity of $3$-coloring paths is $\Theta(\log^\ast n)$.  The existence of the complexity class $\Theta(\log n)$ for  $\BBB$ and $\CCC$  is due to the well-known that the complexity of  $\Delta$-coloring trees is $\Theta(\log n)$. All the remaining results are due to \cref{thm-path-LCL-1,thm-path-LCL-2,thm-path-LCL-3,thm-path-LCL-4,thm-time-hierarchy-main}. 
 
 \paragraph{The conjecture.} For the sake of comparison, here we also illustrate the complexity landscapes described in \cref{conjecture-main}.
 \begin{align*}
   &\AAA_0   && O\left(1\right) \\
  &\AAA_k   && O\left(1\right) 
\frac{\times}{\hspace{0.6cm}} 
\Theta\left(\log^\ast n\right) 
\frac{\times}{\hspace{0.6cm}} 
\Theta\left(n^{\frac{1}{k}}\right) 
\frac{\times}{\hspace{0.6cm}} 
\Theta\left(n^{\frac{1}{k-1}}\right) 
\frac{ \times}{\hspace{0.25cm}} 
\cdots
\frac{ \times}{\hspace{0.25cm}} 
\Theta\left({n^{\frac{1}{2}}}\right) 
\frac{\times}{\hspace{0.6cm}} 
\Theta\left(n\right) \\
  &\BBB   && O\left(1\right) 
\frac{\times}{\hspace{0.6cm}} 
\Theta\left(\log^\ast n\right) 
\frac{\times}{\hspace{0.6cm}} 
\Theta\left(\log  n\right) 
\frac{ \times}{\hspace{0.25cm}} 
\cdots
\frac{ \times}{\hspace{0.25cm}} 
\Theta\left(n^{\frac{1}{3}}\right) \frac{\times}{\hspace{0.6cm}} 
\Theta\left({n^{\frac{1}{2}}}\right)
\frac{\times}{\hspace{0.6cm}} 
\Theta\left(n\right) \\
  &\CCC   && O\left(1\right) 
\frac{\times}{\hspace{0.6cm}} 
\Theta\left(\log^\ast n\right) 
\frac{\times}{\hspace{0.6cm}} 
\Theta\left(\log  n\right) 
\frac{\text{dense}}{\hspace{1cm}} 
\Theta\left(n\right) 
 \end{align*}

\subsection{Roadmap}
In \cref{sect-bounded-pathwidth}, we prove that $\bigcup_{0 \leq i < \infty} \AAA_i$ is a classification of all bounded-pathwidth minor-closed graph classes. 
In \cref{sect:bounded-growth}, we show a combinatorial characterization of the set of graph classes $\AAA_i$ based on the growth rate of the size of the $d$-radius neighborhood.
In \cref{sect:complexity-landscape}, we use the combinatorial characterization to prove    \cref{thm-path-LCL-1,thm-path-LCL-2,thm-path-LCL-3,thm-path-LCL-4}.
In \cref{sect:decide}, we present a polynomial-time algorithm that decides whether $\GG \in \AAA_i$ for any given minor-closed graph class $\GG$.
In \cref{sect-dense-region}, we establish the dense region $[\Theta(\log n), \Theta(n)]$ for any minor-closed graph class that has unbounded treewidth.

\section{A Classification of Bounded-Pathwidth Networks}\label{sect-bounded-pathwidth}

In this section, we prove \cref{lem-pathwidth-partition,thm-pathwidth-k}. That is,  $\bigcup_{0 \leq i < \infty} \AAA_i$ is a classification of all bounded-pathwidth minor-closed graph classes, and $\AAA_k$ contains the class of graphs of pathwidth at most $k$. We need the following well-known characterization of bounded-pathwidth minor-closed graph classes by Robertson and Seymour~\cite{ROBERTSON1983forest}.

\begin{theorem}[{Excluding forest theorem~\cite{ROBERTSON1983forest}}]\label{thm-excluding-forest}
A minor-closed graph class $\GG$ has bounded pathwidth if and only if $\GG$ does not contain all forests.
\end{theorem}

\cref{lem-pathwidth-partition,thm-pathwidth-k} are proved using the observation that the pathwidth of $T_{i,s}$ equals $i$ whenever $s \geq 3$. The calculation of the pathwidth of $T_{i,s}$ can be done in a way similar to the folklore calculation to show that 
the pathwidth of a complete ternary tree of height $d$ is precisely $d$. We still present a proof of the observation here for the sake of completeness. 

\begin{observation}\label{lem-pathwidth-calculation}
 For any two integers $i$ and $s$ such that $i \geq 1$ and $s \geq 3$, the pathwidth of $T_{i,s}$ is $i$.
\end{observation}

We need an auxiliary lemma.

\begin{lemma}[\cite{Scheffler1990}]\label{lem-pathwidth-calculation-aux}
Consider a rooted tree $T$ whose root $r$ has three children $u_1$, $u_2$, and $u_3$. If the pathwidth of the subtree rooted at $u_i$ is at least $k$ for each $i \in\{1,2,3\}$, then the pathwidth of $T$ is at least $k+1$.
\end{lemma}

We are now ready to prove \cref{lem-pathwidth-calculation}.

\begin{proof}[Proof of \cref{lem-pathwidth-calculation}]
We first prove the upper bound. For the base case, $(X_1, X_2, \ldots, X_{s-1})$ with $X_i = \{v_i, v_{i+1}\}$ 
is a path decomposition of the $s$-node path $(v_1, v_2, \ldots, v_s)$, showing that the pathwidth of $T_{i,s}$ is at most $1$.

Given that the pathwidth of  $T_{i-1, s}$ is at most $i-1$, we show that the pathwidth of $T_{i, s}$ is at most $i$.
Let the $s$-node path $(v_1, v_2, \ldots, v_s)$ and $s$ instances $T_1, T_2, \ldots, T_s$ of  $T_{i-1, s}$ be the ones in the definition of $T_{i,s}$. For each $1 \leq j \leq s$, let $(X_{1,j}, X_{2,j}, \ldots, X_{k,j})$ be a path decomposition of $T_j$ of width at most $i-1$. Define $X_{l,j}' = X_{l,j} \cup \{v_j\}$. Then  
\[\left(X_{1,1}', X_{2,1}', \ldots, X_{k,1}', \{v_1, v_2\}, X_{1,2}', X_{2,2}', \ldots, X_{k,2}',  \{v_2, v_3\}, \ldots, \{v_{s-1}, v_s\}, X_{1,s}', X_{2,s}', \ldots, X_{k,s}'\right)\]
is a path decomposition of  $T_{i, s}$ of width at most $i$.

For the rest of the proof, we consider the lower bound. It suffices to consider the case of $s = 3$, as the pathwidth of $T_{i,s}$ is at least the pathwidth of $T_{i,3}$ for each $s \geq 3$. For the base case, it is trivial that the pathwidth of  $T_{1,3}$ is at least $1$.

Given that the pathwidth of  $T_{i-1, 3}$ is at least $i-1$, we show that the pathwidth of $T_{i, 3}$ is at least $i$. Consider the path $(v_1, v_2, v_3)$ in the definition of $T_{i, 3}$. We re-root the tree by setting $v_2$ as the root. Now $v_2$ has three children. Let $T$ be any subtree of  $T_{i, 3}$ rooted at one of the children of $v_2$. Observe that $T$ contains $T_{i-1, 3}$ as a subgraph, so the pathwidth of $T$ is at least $i-1$ by induction hypothesis. Applying \cref{lem-pathwidth-calculation-aux} to $T_{i,3}$ rooted at $v_2$, we obtain that the pathwidth of  $T_{i, 3}$ is at least $i$.
\end{proof}

Using \cref{lem-pathwidth-calculation}, we now prove \cref{lem-pathwidth-partition,thm-pathwidth-k}.

\thmpwa*
\begin{proof}
First of all, \cref{thm-excluding-forest} implies that any minor-closed graph class $\GG \notin \AAA$ cannot belong to $\AAA_i$ for any $i$, as $\GG$ contains all forests, so $\bigcup_{0 \leq i < \infty} \AAA_i \subseteq \AAA$.

We claim that $\AAA_1, \AAA_2, \ldots$ are disjoint sets. Suppose there are two indices $i$ and $j$ such that $\GG \in \AAA_i$ and $\GG \in \AAA_j$ and $i < j$. Then we have $T_{j', s} \in \GG$ for all  $j'$ and $s$ such that $1 \leq j' \leq j$ and $s \geq 1$ by (C1) in \cref{def-pw}, as $T_{j', s}$ is a minor of $T_{j, s}$.
However, $\GG \in \AAA_i$ implies that $T_{i+1, s} \notin \GG$ for some $s$ by (C2)  in \cref{def-pw}. This is a contradiction because  $1 \leq i+1 \leq j$.

It remains to show that each graph class $\GG \in \AAA$ belongs to $\AAA_i$ for some index $i$. 
As $\GG$ has bounded pathwidth, let $k < \infty$ be the maximum pathwidth of graphs in $\GG$. Then $T_{k+1, s} \notin \GG$ for all $s \geq 3$ by \cref{lem-pathwidth-calculation}. We pick $i$ to be the smallest index such that $T_{i+1, s} \notin \GG$ for some $s$. Such an index $i$ exists, and we must have $0 \leq i \leq k$.
As a result, $\GG \in \AAA_{i}$, as both (C1) and (C2)  in \cref{def-pw}  are satisfied due to our choice of $i$.
\end{proof}

\thmpwb*
\begin{proof}
Let $\GG$ be the graph class that contains all graphs with pathwidth at most $k$.
We have $T_{k, s} \in \GG$ for all positive integers $s$ by 
 \cref{lem-pathwidth-calculation}, so (C1) in \cref{def-pw} is satisfied. Here we use the fact that $T_{k,s'}$ is a minor of $T_{k,s}$ whenever  $s' \leq s$.
 Since $\GG$  does not contain any graph with pathwidth at least $k+1$, by 
 \cref{lem-pathwidth-calculation}, $T_{k+1, s} \notin \GG$ for all $s \geq 3$, so (C2) in \cref{def-pw}  is satisfied.
\end{proof}

\section{The Bounded Growth Property}\label{sect:bounded-growth}

In this section, we prove an alternative characterization of the set of graph classes $\AAA_i$ that will be crucial in the complexity-theoretic study in \cref{sect:complexity-landscape}. 

Let $g\geq 0$. We say that a graph class $\GG$ is \emph{$g$-growth-bounded} if, for every integer $\Delta\geq  3$, there exists a constant $C_\Delta > 0$ only depending on $\Delta$ and $\GG$  such that, for every $d\geq 1$, every $G=(V,E) \in\GG^\Delta$, and every $r\in V$, \[|\{v\in V \ | \  \dist(v,r)\leq d\}|\leq C_\Delta \cdot d^g.\]

In other words, $\GG$ is $g$-growth-bounded if the $d$-radius neighborhood of any node in any bounded-degree graph in $\GG$ has size $O(d^g)$.
Recall that $\GG^\Delta$ is the set of graphs $G \in \GG$ with maximum degree at most $\Delta$.  

The goal of this section is to prove the following result, which gives an alternative characterization of the set $\AAA_k$: It is the set of all minor-closed graph classes $\GG$ such that $k$ is the smallest number such that  $\GG$ is $k$-growth-bounded.

\begin{proposition}\label{thm:alt-definition}
Let $\GG$ be any minor-closed graph class.
If $\GG \notin \bigcup_{0 \leq i < \infty} \AAA_i$, then 
$\GG$ is not $k$-growth-bounded for any $k < \infty$.
If $\GG \in \AAA_k$, then $k$ is the smallest number such that  $\GG$ is $k$-growth-bounded.
\end{proposition}

We first prove an easy part of   \cref{thm:alt-definition}.

\begin{lemma}\label{lem:alt-onlyif}
Let $\GG$ be any minor-closed graph class such that $\GG \notin \AAA_0 \cup \AAA_1 \cup \cdots \cup \AAA_k$. For any $s < k+1$, $\GG$ is not $s$-growth-bounded.
\end{lemma} 
\begin{proof}
Consider the set of graphs $\SSS = \{T_{k+1,x} \ | \ x \geq 3\}$. By \cref{def-pw}, as $\GG \notin \AAA_0 \cup \AAA_1 \cup \cdots \cup \AAA_k$,  $\SSS$ is a subset of $\GG^\Delta$ for each $\Delta \geq 3$. From the graph structure of $T_{k+1,x}$, the size of the $d$-radius neighborhood of a node in a graph in $\SSS$ cannot be bounded by any function $O(d^s)$ with $s < k+1$.
\end{proof}

\paragraph{Layers of nodes.} Consider the following terminology.  Given a positive integer $C$ and a rooted tree $T=(V,E)$, we define a sequence $L_0^\ast, L_1, L_1^\ast, L_2, L_2^\ast, \ldots$ of node subsets  of $T$ as follows.
\begin{itemize}
    \item $L_0^\ast = V$ is the set of all nodes in $T$. 
    \item For each $i \geq 1$, $L_i$ is the set of nodes $v$ such that the subtree rooted at $v$ contains at least $C$ nodes in $L_{i-1}^\ast$.
    \item  For each $i \geq 1$, $L_i^\ast$ is the set of nodes $v$  having at least two children in $L_i$.
\end{itemize}
For each non-negative integer $k$, we say that a rooted tree $T$ is \emph{$(k,C)$-limited} if  $L_{k+1} = \emptyset$. 

We write $N^d(r)$ to denote the set of nodes $v \in V$ with $\dist(v,r) \leq d$. We have the following auxiliary lemma.

\begin{lemma}\label{lem-neighborhood-size-aux}
If $T$ is a $(k,C)$-limited rooted tree whose maximum degree is at most $\Delta$, then $|L_{k} \cap N^d(r)| \leq \Delta C d$, where $r$ is the root.
\end{lemma}
\begin{proof} 
Since $T$ is  $(k,C)$-limited, we have $L_{k+1} = \emptyset$, so $|L_{k}^\ast| < C$. If we remove the edges between each $v \in L_{k}^\ast$ and its children in the rooted tree $T$, then the set of nodes in $L_{k}$ are partitioned into paths $P_1, P_2, \ldots, P_x$, as $L_{k}^\ast$ is the set of nodes in $L_{k}$ with at least two children in $L_k$. Here $x$ is the number of paths in the partition. Orienting each edge toward the parent, we write each path $P_i$ as $P_i = u_{i,1} \leftarrow u_{i,2} \leftarrow \cdots$. 

Observe that $u_{i,1}$ is either the root $r$ of $T$ or a child of a node in $L_k$. Since the paths  $P_1, P_2, \ldots, P_x$ are disjoint, at most one path $P_i$ has $u_{i,1} = r$. The number of nodes whose parent is in $L_k^\ast$ is at most $|L_k^\ast|(\Delta -1) + 1$. Hence the total number $x$ of paths is at most $x = |L_k^\ast|(\Delta -1) + 2  \leq (C-1)(\Delta-1)+2 \leq C\Delta$.

A necessary condition for a node $v \in L_k$ to be in $N^d(r)$ is that it is within the first $d$ nodes of some path $P_i$, we have $|L_{k} \cap N^d(r)| \leq xd \leq \Delta C d$.
\end{proof}

Using \cref{lem-neighborhood-size-aux}, we show that the size of the $d$-radius neighborhood of the root $r$ of any $(k,C)$-limited rooted tree $T$  is small.

\begin{lemma}\label{lem-neighborhood-size}
For any three integers $C \geq 1$, $k \geq 0$, and $\Delta \geq 1$, there exists a constant $K_{C,k,\Delta}$ only depending on $C$, $k$, and $\Delta$ such that the number of nodes within distance $d$ to the root in any $(k,C)$-limited rooted tree $T$ with degree at most $\Delta$ is upper bounded by  $K_{C,k,\Delta} \cdot d^k$.
\end{lemma}
\begin{proof} 
We set $K_{C,k,\Delta}$ as follows.
\[
K_{C,k,\Delta}
= 
\begin{cases}
C-1, & \text{if $k=0$,}\\
\Delta C \cdot (\Delta K_{C,k-1,\Delta} + 1)  & \text{if $k \geq 1$}.
\end{cases}
\]

Let $T=(V,E)$ be any  $(k,C)$-limited rooted tree, so $L_{k+1} = \emptyset$. Let $r$ be the root of $T$.  We verify that the number of nodes within distance $d$ to the root in $T$ is at most $|N^d(r)| \leq K_{C,k,\Delta} \cdot d^k$.

\paragraph{Base case.}
For the case of $k = 0$, if $L_1 = \emptyset$, then $T$ contains at most $C-1$ nodes, since otherwise the root of $T$ belongs to $L_1$. Therefore, $|N^d(r)| \leq |V| \leq C-1 = K_{C,0,\Delta}$.

\paragraph{Inductive step.}
Now, suppose $k \geq 1$.
If $L_{k} = \emptyset$, then $T$ is $(C,k-1)$-limited, so the number of nodes within distance $d$ to the root is at most $K_{C,k-1,\Delta} d^{k-1} \leq K_{C,k,\Delta} d^k$ by inductive hypothesis.

Suppose $L_{k} \neq \emptyset$, then the root $r$ belongs to $L_{k}$, and the number of nodes in $L_k$ within distance $d$ to the root is at most $\Delta C d$ by \cref{lem-neighborhood-size-aux}. Let $S = N^d(r) \cap L_k$ denote these nodes.

Each node in $N^d(r) \setminus S$  belongs to the tree rooted at a child $u$ of some $v \in S$. Let $S'$ be the set of nodes in $V \setminus L_k$ that are children of nodes in $S$. Then $|S'| \leq \Delta |S|$.

Since $u \notin L_k$ for each $u \in S'$, the subtree of $T$ rooted at $u$ is $(C,k-1)$-limited, so it contains at most $K_{C,k-1,\Delta} d^{k-1}$ nodes that are within distance $d$ to $u$ by inductive hypothesis.
Therefore, 
\begin{align*}
|N^d(r)| &= |N^d(r) \cap L_k| + |N^d(r) \setminus L_k| \\
&\leq |S| + |S'| \cdot K_{C,k-1,\Delta} d^{k-1}\\
&\leq \Delta C d  + \Delta^2 C d \cdot K_{C,k-1,\Delta} d^{k-1}\\
& \leq \Delta C \cdot (1 + \Delta K_{C,k-1,\Delta}) d^{k}\\
& = K_{C,k,\Delta} d^{k},
\end{align*}
as desired.
\end{proof}

\paragraph{Rooted minors.} A \emph{rooted} graph is a graph $G$ with a distinguished root node $r \in V$. A rooted graph can also be treated as a graph without a root. We say that $H$ is a \emph{rooted minor} of $G$ if there exist  a partition of $V(G)$ into $k=|V(H)|$ disjoint connected clusters $\CC = \{V_1, V_2, \ldots, V_k\}$ and a bijection between $\CC = \{V_1, V_2, \ldots, V_k\}$ and $V(H)$  satisfying the following two conditions:
\begin{itemize}
    \item For each edge $e$ in $H$, the two clusters in  $\CC$ corresponding to the two endpoints of $e$ are adjacent in $G$. 
    \item The cluster in $\CC$ corresponding to the root of $H$ contains the root of $G$.
\end{itemize}
Observe that in \cref{def-trees}, $T_{i,s}$ is a rooted minor of $T_{i', s'}$  whenever $i \leq i'$ and $s \leq s'$.

We show that a rooted minor isomorphic to $T_{i+1,s}$ exists in any rooted tree $T$ that is not $(i,C)$-limited, for some sufficiently large $C$.

\begin{lemma}\label{lem-minor-aux}
For any  positive integers $x$ and $\Delta$, there is a number $C_{x,\Delta} > 0 $  only depending on $x$ and $\Delta$ such that for each positive integer $i$, any rooted tree $T$ with degree at most $\Delta$ that is not $(i,C_{x,\Delta})$-limited  contains $T_{i+1,x}$ as a rooted minor.
\end{lemma}
\begin{proof}
We define $B_1 = 1$ and $B_{j} = 1 + \Delta \cdot B_{j-1}$ for $j > 1$. Remember that if $T$ is not $(i,C)$-limited, then $L_{i+1} \neq \emptyset$, so $|L_{i}^\ast| \geq C$.
We claim that for $C = B_{x}$, in any  rooted tree $T$ with $|L_{i}^\ast| \geq C$, there is a path $P=(u_1, u_2, \ldots, u_k)$ meeting the following conditions.
 \begin{itemize}
     \item The first node $v_1$ of the path is the root $r$ of $T$.
     \item The path contains $x$ nodes in $L_{i}^\ast$. Denote these nodes as $v_1, v_2, \ldots, v_x$.
     \item All nodes in the path belong to $L_{i-1}$.
 \end{itemize}
 
Such a path $P$ can be constructed as follows. We start with $u_1 = r$. Once $(u_1, u_2, \ldots, u_j)$ has been constructed, if the current path still does not contain  $x$ nodes in $L_{i}^\ast$, then we extend the path by picking $u_{j+1}$ as a child of $u_j$ that maximizes the number of nodes in $L_{i}^\ast$ in the subtree rooted at $u_{j+1}$. 

To analyze this construction, let $\Phi_j > 0$ be the number of nodes in $L_{i}^\ast$ in the subtree rooted at $u_j$. There are two cases.
\begin{itemize}
    \item If $u_j \in L_{i} \setminus L_i^\ast$, then $u_j$ has exactly one child  $u_{j+1}$ in $L_{i}$ with $\Phi_{j+1} = \Phi_j$.
    \item Otherwise, $u_j \in L_i^\ast$. Then $u_j$ has at most $\Delta$ children in $L_{i}$, so $\Phi_{j+1} \geq (\Phi_j - 1)/\Delta$.
\end{itemize}
Our choice of $C = B_{x}$ ensures that this process leads to a path containing $x$ nodes in $L_{i}^\ast$.

Now, consider such a path $P$. If $i = 0$, then $P$ already contains $T_{1,x}$ as a rooted minor, as $P$ contains at least $x$ nodes and $T_{1,x}$ is an $x$-node path.  

From now on, we assume $i > 1$. Since  $v_j \in L_{i}^\ast$, it has a child $w_j \in L_{i}$ outside of $P$. By inductive hypothesis, the subtree $T_j$ rooted at $w_j$ has a rooted minor isomorphic to $T_{i, C}$. For each $1 \leq j \leq x$, we consider the partition of the nodes in $T_j$ witnessing the fact that it has $T_{i, C}$ as a rooted minor. The rest of the nodes in $T$ can be partitioned into $x$ connected clusters $S_1, S_2, \ldots, S_x$ so that $v_j \in S_j$. We must have $r \in S_1$, and the overall partition of the nodes of $T$ shows that $T_{i+1, C}$ is a rooted minor of $T$, so we may set $C_{x, \Delta} = C = B_x$.
\end{proof}

For a graph class  $\GG$, we say that   $\GG$ is \emph{$(k,C)$-limited} if all trees 
$T=(V,E)$ in $\GG$ are $(k,C)$-limited for all choices of the root $r \in V$. It is still allowed that $\GG$ contains graphs that are not trees.

\begin{lemma}\label{lem-aux-limited}
Let $\GG$ be any minor-closed graph class such that for each positive integer $\Delta$ there is a constant $C > 0$ such that $\GG^\Delta$ is $(k,C)$-limited. Then $\GG$ is $k$-growth-bounded.
\end{lemma}
\begin{proof}
Suppose $\GG^\Delta$ is $(k,C)$-limited.
 For each $G=(V,E) \in \GG$ with degree at most $\Delta$, we pick any node $r \in V$, and let $T$ be the tree corresponding to any BFS starting from $r$. 
 Because $\GG$ is minor-closed, we have $T \in \GG$, so  $T$ rooted at $r$ is $(k,C)$-limited.
 Applying  \cref{lem-neighborhood-size} to the rooted tree $T$ with the parameters $C$,$k$, and $\Delta$, we conclude that the size of the $d$-radius neighborhood of $r$ in $G$ is at most $K_{C,k,\Delta} \cdot d^k$, so $\GG$ is $k$-growth-bounded with $C_{\Delta} = K_{C,k,\Delta}$.
\end{proof}

\begin{lemma}\label{lem:alt-if}
If $\GG \in \AAA_0 \cup \AAA_1 \cup \cdots \cup \AAA_k$, then for each integer $\Delta \geq 3$, there exists a constant $C > 0$ such that $\GG^\Delta$ is  $(k,C)$-limited.
\end{lemma} 
\begin{proof}
Suppose there exists $\Delta \geq 3$ such that $\GG^\Delta$ is not $(k,C)$-limited  for all $C <\infty$. We apply \cref{lem-minor-aux} to any tree $T \in \GG^{\Delta}$ that is not $(k,C)$-limited for  $C = C_{x, \Delta}$ for some choice of the root $r$. Then we obtain that $T$ contains $T_{k+1,x}$ as a minor, so $T_{k+1,x} \in \GG$. Since this holds for all $x$, $\GG$ contains $T_{k+1,x}$ for all $x$, implying that $\GG \notin \AAA_0 \cup \AAA_1 \cup \cdots \cup \AAA_k$, which is a contradiction to the assumption that $\GG \in \AAA_0 \cup \AAA_1 \cup \cdots \cup \AAA_k$.  Hence for each integer $\Delta \geq 3$, there exists a constant $C > 0$ such that $\GG^\Delta$ is  $(k,C)$-limited. 
\end{proof}

Now we are ready to prove \cref{thm:alt-definition}.

\begin{proof}[Proof of \cref{thm:alt-definition}]
If $\GG \notin \bigcup_{0 \leq i < \infty} \AAA_i$, then  \cref{lem-pathwidth-partition,thm-excluding-forest} imply that  $\GG$  contains all trees. By considering complete trees, we infer that  $\GG$  is not $k$-growth-bounded for any $k < \infty$.

Now suppose $\GG \in \AAA_k$.
Since $\GG \in  \AAA_0 \cup \AAA_1 \cup \cdots \cup \AAA_k$,  \cref{lem:alt-if,lem-aux-limited} implies that $\GG$ is $k$-growth-bounded. 
Since $\GG \notin \AAA_0 \cup \AAA_1 \cup \cdots \cup \AAA_{k-1}$,  \cref{lem:alt-onlyif} implies that $\GG$ is not $s$-growth-bounded for any $s < k$. Hence   $k$ is the smallest number such that $\GG$ is $k$-growth-bounded. 
\end{proof}

\section{The Complexity Landscape}\label{sect:complexity-landscape}

In this section, we consider the complexity of $\LCL$ problems for the graph class $\GG^\Delta$ with $\GG \in \AAA_k$, for some $k \geq 0$.
We first consider the case of $\GG \in \AAA_0$.

\thmcomplexitya*
\begin{proof}
By \cref{thm:alt-definition}, there is a number $C_\Delta$ such that the size of $d$-radius neighborhood of any $r \in V$ in any $G=(V,E) \in \GG^\Delta$  is at most $C_\Delta$. As the statement holds for all $d$, we infer that each graph in $\GG^\Delta$ has at most $C_\Delta$ nodes. Therefore, if an $\LCL$ problem is solvable in $G \in \GG^\Delta$, then it can be solved in $C_\Delta = O(1)$ rounds.
\end{proof}

For the rest of the section, we consider $k > 0$.


\thmcomplexityb*
\begin{proof}
The $\LCL$ problem $\PP_s$ considered by Chang and Pettie~\cite{ChangP17} can be solved in $O(n^{1/s})$ rounds on \emph{all graphs}, including the ones in $\GG$. It was shown in~\cite{ChangP17} that $\PP_s$  requires $\Omega(n^{1/s})$ rounds to solve in the graphs $T_{s,x}$, for all $x$. Since any $\GG \in \AAA_k$ contains $T_{s,x}$ for all $x$ by the definition of $\AAA_k$, the complexity of $\PP_s$ in any $\GG \in \AAA_k$ is $\Omega(n^{1/s})$. Combining the upper and lower bounds, we conclude that the complexity of $\PP_s$ is $\Theta(n^{1/s})$ in $\GG$.
\end{proof}


\thmcomplexityc*
\begin{proof}
This lemma follows from a modification of the proof of the existence of the deterministic $\omega(\log^\ast n)$ -- $o(\log n)$ complexity gap on general graphs by Chang, Kopelowitz, and Pettie~\cite{ChangKP16}. 

We briefly review their proof and describe the needed modification.
Fix any $\GG \in \AAA_k$.
Consider any deterministic algorithm $\Algo$ solving an $\LCL$ problem $\PP$ in $T(n) = o(n^{1/k})$ rounds for all graphs in $\GG^\Delta$. The goal is to design an algorithm $\Algo'$ that also solves  $\PP$  for all graphs in $\GG^\Delta$ and costs only $O(\log^\ast n)$ rounds.

Suppose the locality radius of $\PP$ is $r$. If we can assign distinct $O(\log \tilde{n})$-bit identifiers to each node such that the $(T(\tilde{n})+r)$-radius neighborhood of each node has size at most $\tilde{n}$ and does not contain repeated identifiers, then $\PP$ can be solved in $T(\tilde{n})$ rounds by running $\Algo$ with these  $O(\log \tilde{n})$-bit identifiers and pretending that the underlying network has $\tilde{n}$ nodes.
To see that this strategy works, we use the property that $\GG$ is minor-closed. We show that the output label of each node $v$ resulting from the above approach is correct. Given any node $v$ in $G$, consider a subgraph $\tilde{G}$ of $G$ that contains exactly $\tilde{n}$ nodes and contains the entire $(T(\tilde{n})+r)$-radius neighborhood of $v$. We assign distinct  $O(\log \tilde{n})$-bit identifiers to the nodes in $\tilde{G}$ in such a way that the $O(\log \tilde{n})$-bit identifiers in the $(T(\tilde{n})+r)$-radius neighborhood of $v$ are chosen to be the same in $G$ and $\tilde{G}$.
Because $\GG$ is minor-closed, we have $\tilde{G} \in \GG$.
The output labels resulting from running $\Algo$ in $G$ and $\tilde{G}$ must be the same in the $r$-radius neighborhood of $v$, so the correctness of $\Algo$ in $\tilde{G}$ implies that the output label of $v$ must be locally correct in $G$.

We show that we can choose $\tilde{n} = O(1)$ to satisfy the property that the $(T(\tilde{n})+r)$-radius neighborhood of each node has size at most $\tilde{n}$.
Let $d= T(\tilde{n})+r$.
Then $d = o(\tilde{n}^{1/k})$ because $r = O(1)$ and $T(\tilde{n}) = o(\tilde{n}^{1/k})$.
By \cref{thm:alt-definition}, the size of the $d$-radius neighborhood of any node in any graph $G \in \GG^\Delta$ is $O(d^k)$. Thus,  the size of the $(T(\tilde{n})+r)$-radius neighborhood of each node is upper bounded by $o(\tilde{n})$, meaning that by choosing $\tilde{n}$ to be a sufficiently large constant, the size of the  $(T(\tilde{n})+r)$-radius neighborhood of each node is at most $\tilde{n}$.

As $\tilde{n} = O(1)$ is a constant, the needed $O(\log \tilde{n})$-bit  identifiers can be computed in $O(\log^\ast n)$ rounds deterministically using the coloring algorithm of~\cite{FraigniaudHK16}. Since $T(\tilde{n}) = O(1)$ is also a constant, this approach yields a new algorithm $\Algo'$ that solves $\PP$ in just $O(\log^\ast n)$ rounds.
\end{proof}

\thmcomplexityd*
\begin{proof}
This lemma is proved by the Ramsey-theoretic technique of Naor and Stockmeyer~\cite{NaorS95}, as discussed in the paper of Chang and Pettie~\cite{ChangP17}. Although this proof considers only deterministic algorithms, the result applies to randomized algorithms as well, because it is known that randomness does not help for algorithms with round complexity $2^{O(\log^\ast n)}$, as shown by Chang, Kopelowitz, and Pettie~\cite{ChangKP16}.

Fix a graph class $\GG^\Delta$ with  $\GG \in \AAA_k$.
Consider a $T(n)$-round deterministic algorithm $\Algo$ solving an $\LCL$ problem $\PP$ for a graph class $\GG^\Delta$. Fix a network size $n$. Consider the following parameters.
\begin{itemize}
    \item $\tau = T(n)$ is the round complexity of  $\Algo$ on $n$-node graphs in $\GG^\Delta$.
    \item $p$ is the maximum number of nodes in a $\tau$-radius neighborhood of a node in a graph in $\GG^\Delta$. According to \cref{thm:alt-definition}, $p = O(\tau^k)$.
    \item $m$ is the maximum number of nodes in a $(\tau+r)$-radius neighborhood of a node in a graph in $\GG^\Delta$.
    \item $c$ is the number of distinct functions mapping each possible $\tau$-radius neighborhood of a graph in $\GG^\Delta$,  whose nodes are equipped with distinct labels drawn from some set $S$ with size $p$, to an output label of $\PP$.
\end{itemize}

Let $R(p,m,c)$ be the minimum number of nodes guaranteeing that any edge coloring of a complete $p$-uniform hypergraph with $c$ colors contains a monochromatic clique of size $m$. It is known that $\log^\ast R(p,m,c) \leq p + \log^\ast m + \log^\ast c + O(1)$. 
Suppose the space of unique identifiers has size $n^C$. As long as $n^C \geq R(p,m,c)$, it is possible to transform $\Algo$ into an $O(1)$-round deterministic algorithm solving $\Algo$. See~\cite{ChangP17,NaorS95} for details. 

Since the part $\log^\ast m + \log^\ast c + O(1)$ is negligible comparing with $p$,  if $p \ll \log^\ast n$, then the above transformation is possible. Since $p = O(\tau^k)$, the inequality $p \ll \log^\ast n$ holds whenever the round complexity of $\Algo$ is $o\left( (\log^\ast n)^{1/k}\right)$. Hence  the complexity of $\PP$ is not within $\omega(1)$ and $o\left((\log^\ast n)^{1/k}\right)$ on $\GG^\Delta$.
\end{proof}

\section{The Computational Complexity of the Classification} 
\label{sect:decide}

In this section, we prove \cref{lem-Ak-decide} by showing a polynomial-time algorithm that decides whether $\GG \in \AAA_i$ for any given minor-closed graph class $\GG$, where $\GG$ is represented by a finite list of excluded minors $H_1, H_2, \ldots, H_k$.

\begin{lemma}\label{lem-path-like-decide-aux}
Let $\GG$ be the class of $\{H_1, H_2, \ldots, H_k\}$-minor-free graphs. Then $\GG \in \AAA_0 \cup \AAA_1 \cup \cdots \cup \AAA_i$ if and only if there exist an integer $s \geq 3$ and an index $j \in [k]$  such that $H_j$ is a minor of $T_{i+1,s}$.
\end{lemma}
\begin{proof}
Suppose  $H_j$ is a minor of $T_{i+1,s}$ for some $s \geq 3$. Then  $T_{i+1,s} \notin \GG$. For each $x \geq i+1$, $T_{i+1,s}$ is a minor of $T_{x,s}$, so $T_{x,s} \notin \GG$, as $\GG$ is closed under minors. Since a necessary condition for $\GG \in \AAA_x$ is $T_{x,s} \in \GG$, we must have $\GG \notin \AAA_x$. Since $H_j$ is a tree, $\GG$ has bounded pathwidth by \cref{thm-excluding-forest}. Therefore, $\GG \in \AAA \setminus(\AAA_{i+1} \cup \AAA_{i+2} \cup \cdots) = \AAA_0 \cup \AAA_1 \cup \cdots \cup \AAA_i$.

For the other direction, suppose $\GG \in \AAA_0 \cup \AAA_1 \cup \cdots \cup \AAA_i$. Then  $\GG \in \AAA_x$ for some $x \in [0, i]$. By the definition of $\AAA_x$, there is an index $s \geq 3$ such that $\GG$ does not contain $T_{x+1,s}$. Since  $T_{x+1,s}$ is a minor of  $T_{i+1,s}$, $\GG$ also does not contain $T_{i+1,s}$, so there is some graph $H_j$ in the list of excluded minors $\{H_1, H_2, \ldots, H_k\}$ such that $H_j$ is a minor of $T_{i+1,s}$.
\end{proof}

\begin{lemma}\label{lem-Ak-decide-aux}
For any rooted tree $H=(V,E)$, the following two statements are equivalent.
\begin{itemize}
    \item There exists a positive integer $s$  such that $H$ is a rooted minor of $T_{j,s}$.
    \item There is a  choice of  a node $u$ in $H$ satisfying the following requirement. Let $P$ be the unique path connecting $u$ and the root $r$ of $H$. Let $S$ be the set of nodes in $V \setminus P$ whose parent is in $P$. For each $v \in S$, there exists a positive integer $s'$  such that the subtree rooted at $v$ is a rooted minor of $T_{j-1,s'}$.
\end{itemize}
\end{lemma}
\begin{proof}
Suppose $H=(V,E)$ is a rooted minor of $T_{j,s}$. Consider any clustering $\CC$ of the node set of $T_{j,s}$ witnessing the fact that $T_{j,s}$ contains $H$ as a rooted minor.
Consider the path $(v_1, v_2, \ldots, v_s)$ and the $s$ rooted trees $T_1, T_2, \ldots, T_s$ in the definition of $T_{j,s}$. Let $\CC^\ast$ be the set of clusters in $\CC$ containing a node in $\{v_1, v_2, \ldots, v_s\}$. Then the set of nodes in $H$ corresponding to the clusters in $\CC^\ast$ form a path $P$ connecting the root $r$ of $H$ to some other node $u$.

Let $S$ be the set of nodes in $V \setminus P$ whose parent is in $P$. We show that the subtree rooted at each $v \in S$ is a rooted minor of $T_{j-1,s}$.
For each node $v \in S$, let $\CC_v$ be the set of clusters in $\CC \setminus \CC^\ast$ corresponding to nodes in the subtree of $H$ rooted at $v$. The union of the clusters in $\CC_v$ must be a connected set $U$ of nodes that are confined to one $T_i$ of the rooted trees $T_1, T_2, \ldots, T_s$ in the definition of $T_{j,s}$. Let $T'$ be the rooted subtree of $T_i = T_{j-1,s}$ induced by  $U$. Observe that the cluster in $\CC_v$ corresponding to $v$ contains the root of $T'$. We extend the clustering $\CC_v$ to cover all nodes in $T_i = T_{j-1,s}$ as follows. For each node $w$ in $T_i$ that is not covered by the clustering $\CC_v$, let $w$ join the cluster in $\CC_v$ that contains the unique node $w'$ in $T'$ minimizing $\dist(w,w')$ in $T_i$. The resulting clustering shows that the subtree rooted at $v$ is a rooted minor of $T_i = T_{j-1,s}$.

For the other direction, suppose there is a choice of a node $u$ in $H$ satisfying the following condition. Let $P = (r = u_1, u_2, \ldots, u_x = u)$ be the unique path connecting the root $r$ and $u$ in $H$, where $x$ is the number of nodes in $P$. Let $S$ be the set of nodes in $V \setminus P$ whose parent is in $P$. For each node $v \in S$, there exists a positive integer $s' = s_v$  such that the subtree rooted at $v$ is a rooted minor of $T_{j-1,s'}$.

We choose $s$ to be the maximum of $|S|$ and $s_v$ among all $v \in S$.
We show that $H$ is a rooted minor of $T_{j,s}$ by demonstrating a clustering of the nodes of $T_{j,s}$ meeting all the needed requirements in the definition of rooted minor. Consider the path $(v_1, v_2, \ldots, v_s)$ and the $s$ rooted trees $T_1, T_2, \ldots, T_s$ in the definition of $T_{j,s}$. We partition the path $(v_1, v_2, \ldots, v_s)$ into $x$ subpaths $(v_1, v_2, \ldots, v_s) = P_1 \circ P_2 \circ \cdots \circ P_x$ such that the number of nodes in $P_i$ is at least the number of children of $u_i$ in $S$. This partition exists due to our choice of $s$. 

A clustering of the nodes of $T_{j,s}$ is constructed as follows.
For each child $w$ of $u_i$ such that $w \in S$, we associate $w$ with a distinct index $l_w$ such that $v_{l_w} \in P_i$. Recall that the subtree rooted at $w \in S$ is a rooted minor of $T_{j-1,s_w}$. Since $s \geq s_w$, $T_{j-1,s_w}$ is a rooted minor of $T_{j-1,s}$, so the subtree rooted at $w \in S$ is also a rooted minor of $T_{l_w} = T_{j-1,s}$. We find a clustering of the nodes in $T_{l_w}$ witnessing that the subtree rooted at  $w$ is a rooted minor of $T_{l_w}$.
For each $u_i \in P$, its corresponding cluster $C$ in $T_{j,s}$ is chosen as follows. Start with the subpath $C = P_i$. For each node $v_{l} \in P_i$ that is not associated with any node in $S$, we add all nodes in the subtree $T_l$ to $C$. 

The clustering covers all nodes in $T_{j,s}$, the root $v_1$ of $T_{j,s}$ belongs to the cluster corresponding to the root $u_1 = r$ of $H$, and each edge in $H$ corresponds to a pair of adjacent clusters, so $H$ is a rooted minor of $T_{j,s}$.
\end{proof}

Now we are ready to prove \cref{lem-Ak-decide}.

\thmdecide*
\begin{proof}
In view of \cref{lem-path-like-decide-aux},  $\GG \in \AAA_i$  if and only if the following two conditions are met.
\begin{itemize}
    \item If $i > 0$, then for all integers $s \geq 1$, each graph $H_j \in\{H_1, H_2, \ldots, H_k\}$ is not a minor of $T_{i,s}$.
    \item There exist an integer $s \geq 1$ and a graph $H_j \in\{H_1, H_2, \ldots, H_k\}$ that is a minor of $T_{i+1,s}$.
\end{itemize}
Therefore, to decide whether $\GG \in \AAA_i$ in polynomial time, it suffices to design a polynomial-time algorithm that accomplishes the following task for a fixed index $j$. 
\begin{itemize}
    \item  Given a graph $H$, decide if there exists an integer $s \geq 1$ such that $H$ is a minor of $T_{j,s}$.
\end{itemize}
 We may assume that $H$ is a tree, since otherwise we can immediately decide that $H$ is not a minor of $T_{j,s}$. Moreover, the above task can be reduced to following rooted version by trying each choice of the root node in $H$.
 \begin{itemize}
    \item  Given a rooted tree $H$, decide if there exists an integer $s \geq 1$ such that $H$ is a rooted minor of $T_{j,s}$.
\end{itemize}

For the rest of the proof, we design a polynomial-time algorithm for this task by induction on $j$.
For the base case of $j = 1$, as $T_{1,s}$ is simply an $s$-node path, the task of checking whether $H$ is a rooted minor of $T_{1,s}$ for some $s$ can be restated as follows. Given a rooted tree $H$, check whether it is a path where its root is an endpoint of the path, that is, $H$ is isomorphic to $T_{1,s}$ for some $s$. This task is doable in polynomial time.  
Assuming that we have a polynomial-time algorithm for $j-1$, the characterization of \cref{lem-Ak-decide-aux} yields a polynomial-time algorithm for $j$ by checking all choices of the node $u$ in the characterization of 
 \cref{lem-Ak-decide-aux}. 
\end{proof}

\paragraph{Remark.} In the \emph{minor containment} problem, we are given two graphs $H$ and $G$, and the goal is to decide whether $H$ is a minor of $G$. The problem is well-known to be polynomial-time solvable when $H$ has constant size~\cite{kawarabayashi2012disjoint,robertson1995graph}. In the minor containment instances in the above proof,  $H$ is one of the forbidden minors $H_1, H_2, \ldots, H_k$. As  $\{H_1, H_2, \ldots, H_k\}$ is seen as the input to the problem considered in \cref{lem-Ak-decide}, these graphs cannot be treated as constant-size graphs, so the polynomial-time algorithms in~\cite{kawarabayashi2012disjoint,robertson1995graph} cannot be adapted here. 

If $H$ does not have constant size, then the problem is NP-complete, even when $G$ is a tree~\cite{gupta1995parallel}. 
A polynomial-time algorithm for minor containment is known for the case where $G$ is a tree and $H$ has bounded degree~\cite{gupta1995parallel}. This polynomial-time algorithm does not apply to our setting in the above proof, as the maximum degree 
 of the graphs $H_1, H_2, \ldots, H_k$ can be arbitrarily large.

\section{Conclusions}\label{sect:conclusions}

In this work, we proposed a conjecture that characterizes the complexity landscape of $\LCL$ problems for any minor-closed graph class, and we proved a part of the conjecture. We believe that our classification  $\AAA = \bigcup_{0 \leq i < \infty} \AAA_i$ of bounded-pathwidth minor-closed graph classes offers the right measure of similarity between a minor-closed graph class and paths from the perspective of locality of distributed computing.

A significant effort is required to completely settle the conjecture. To extend the existing techniques to establish complexity gaps for $\LCL$ problems from paths or trees to bounded-pathwidth graphs or bounded-treewidth graphs, new local algorithms to decompose these graphs are needed. One major challenge here is that both path decomposition and tree decomposition are \emph{global} in that computing them requires at least diameter rounds, so they are not directly applicable to the $\LOCAL$ model. New results along this direction not only have complexity-theoretic implications but it is likely going to yield new algorithmic tools in the $\LOCAL$ model.

On a broader note, we hope our work will inspire others to explore structural graph theory techniques in the study of the algorithms and complexities of local distributed graph problems and build bridges between distributed computing and other areas within algorithms and combinatorics.

\bibliographystyle{alpha}
\bibliography{references}


\appendix 

\section{Unbounded-Treewidth Networks}\label{sect-dense-region}

The goal of this section is to show that for any minor-closed graph class $\GG$ with unbounded treewidth,  its complexity landscape is \emph{dense} in the super-logarithmic region, in a sense similar to the result of Balliu~et~al.~\cite{Balliu18}.  In particular,  we will construct $\LCL$ problems with complexity $\Theta(n^{c})$ for any rational number $c$ such that $0 \leq c \leq 1$ and complexity $\Theta(\log^{c} n)$ for any rational number $c \geq 1$.

 \paragraph{Treewidth.} A \emph{tree decomposition} of a graph $G=(V,E)$ is a tree $T$ whose node set is a family  $\{X_1, X_2, \ldots, X_k\}$ of subsets of $V$ meeting the following conditions.
\begin{description}
    \item[(T1)] $X_1 \cup X_2 \cup \cdots \cup X_k = V$.
    \item[(T2)] For each edge $e=\{u,v\} \in E$, there is a subset $X_i$ that contains both $u$ and $v$. 
    \item[(T3)] If $X_i$ and $X_j$ both contain $v$, then $v$ is also included each $X_l$ in the unique path connecting $X_i$ and $X_j$ in $T$.
\end{description}

The \emph{width} of a tree decomposition is $\max_{1 \leq i \leq k} |X_i| - 1$. The \emph{treewidth} of a graph is the minimum width over all tree decompositions of the graph. Intuitively, the treewidth of a graph measures how similar it is to a tree.

\begin{theorem}[{Excluding grid theorem~\cite{ROBERTSON1986excludingPlanar}}]\label{thm-excluding-grid}
A minor-closed graph class $\GG$ has bounded treewidth if and only if $\GG$ does not contain all planar graphs.
\end{theorem}

The excluding grid theorem of Robertson and Seymour~\cite{ROBERTSON1986excludingPlanar} allows us to make use of all bounded-degree planar graphs when we design our $\LCL$ problems. 
 Similar to~\cite{balliuBOS18}, our construction of $\LCL$ problems is a black-box transformation of a certain type of Turing machine $\MM$ into a distributed problem $\PP$, where the round complexity of $\PP$ is determined by the time complexity of $\MM$. To determine the distributed complexity classes attainable through our construction, we need to formalize the notion of \emph{time constructable functions} with respect to certain Turing machines.

\paragraph{Turing machines.}
A linear-space-bounded Turing machine is described by a 6-tuple $\MM = (Q, \Gamma , \{\bot_\LLLL, \bot_\RRRR\} , F,  q_{0},  \delta)$, as follows.  
\begin{itemize}
    \item $Q$ is a finite set of states.
    \item $\Gamma$ is a finite set of tape symbols.
    \item $\{\bot_\LLLL, \bot_\RRRR\} \subseteq \Gamma$ are two special blank symbols.
    \item $F \subseteq Q$ is a non-empty set of accepting states.
    \item $q_0 \in Q \setminus F$ is the initial state.
    \item $\delta$ is a transition function mapping  $(Q \setminus F) \times \Gamma$ to  $Q \times \Gamma \times \{\LLLL,\RRRR\}$.
\end{itemize}

Here $\delta(q,a) =(q',a',X)$ means that whenever the head of the machine is over a cell whose symbol is $a$ and the current state of the machine is $q$, the machine will write $a'$ to the current cell, transition to the state $q'$, and move the lead to the left cell (if $X = \LLLL$) or the right cell (if $X = \RRRR$). The machine terminates once it enters a state in $F$.

Given an input string $S=(a_1, a_2, \ldots, a_k)$ where each element is in $\Gamma \setminus \{\bot_\LLLL, \bot_\RRRR\}$, we assume that the tape of the machine is initially $(\bot_\LLLL,  a_1, a_2, \ldots, a_k, \bot_\RRRR)$, and the head of the machine is initially over the left-most cell of the tape, whose symbol is $\bot_\LLLL$. 

We do not allow the Turing machine to use additional space. This is enforced by the following rules for any transition $\delta(q,a) =(q',a',X)$.
\begin{itemize}
    \item Whenever $a = \bot_\LLLL$, we must have $X = \RRRR$ and $a' = \bot_\LLLL$.
    \item Whenever $a = \bot_\RRRR$, we must have $X = \LLLL$ and $a' = \bot_\RRRR$.
    \item Whenever $a \in \Gamma \setminus \{\bot_\LLLL, \bot_\RRRR\}$, we must have $a' \in \Gamma \setminus \{\bot_\LLLL, \bot_\RRRR\}$.    
\end{itemize}

That is, $\bot_\LLLL$ always indicates the left border,   $\bot_\RRRR$ always indicates the right border, and no intermediate cell is allowed to use these two symbols.

\paragraph{Time-constructable functions.} We say that $T(s)$ is a time-constructable function with respect to  linear-space-bounded Turing machines if there is a machine $\MM = (Q, \Gamma , \{\bot_\LLLL, \bot_\RRRR\} ,   q_{0}, F, \delta)$ meeting the above requirements such that it takes exactly $T(s)$ steps for the machine $\MM$ to terminate given the length-$s$ unary string $S = 0^s$ as the input.

We say that $T(s)$ is \emph{good} if there exists some finite number $C$ such that   $s \leq T(s) \leq T(s+1) \leq C \cdot T(s)$ holds for all positive integer $s$. 
In \cref{lem-time-hierarchy-aux}, we show that exponential functions are time-constructible with respect to linear-space-bounded Turing machines.  

\begin{lemma}\label{lem-time-hierarchy-aux}
For each rational number $c > 0$, there is a good time-constructible function $T(s) = \Theta(2^{cs})$ with respect to linear-space-bounded Turing machines.
\end{lemma}
\begin{proof}
Let $c > 0$ be any rational number. We write $c = p/q$ where $p$ and $q$ are positive integers. We will construct a linear-space-bounded Turing machine $\MM$ that accepts $S = 0^s$ in $\Theta(2^{sc})$ steps. The construction of $\MM$ has two parts. In the first part, we design a Turing machine that finds a substring $S'$ of $S$ of length $\ceil{s/q}$ in $O(s^2)$ steps, for any given positive integer $q$. In the second part, given an input string $S'$ of length $k = \ceil{s/q}$, we design a Turing machine that terminates in $\Theta(b^k)$ steps, for any given positive integer $b$. Composing these two machines with $b = 2^p$, we obtain the required Turing machine that terminates in $\Theta(2^{cs})$ steps given   $S = 0^s$ as the input.

\paragraph{Division.} Given a string  $S = 0^s$ as the input, we find a substring $S'$ of length $\ceil{s/q}$, as follows. First, we relabel these $s$ cells by $0,1,\ldots, q-1, 0, 1, \ldots, q-1, \ldots$. In other words, we perform a single sweep of the input tape and write the symbol $(i \mod q)$ to the $i$th cell. This takes $O(s)$ steps. After that, we perform a bubble sort to sort these numbers in non-decreasing order in $O(s^2)$ steps. After sorting, the initial substring $0^{\ceil{s/q}}$ is the desired output.

\paragraph{Exponentiation.} Given a string $S' = 0^k$ as the input, we design a procedure that takes $\Theta(b^k)$ steps to finish. There are $k$ cells $c_1, c_2, \ldots, c_k$ storing the string $S'$, where all of them are labeled $0$ initially. Let $c_\LLLL$ be the cell left to $c_1$ and   $c_\RRRR$ be the cell right to $c_k$. \cref{alg:simple} counts from $0$ to $b^k - 1$ using $c_1, c_2, \ldots, c_k$ to represent numbers in the base $b$. The procedure takes $\Theta(b^k)$ steps to finish in a Turing machine. 
\end{proof}

\begin{algorithm}[!ht]
\SetAlgoLined
At the beginning, the head of the machine is over the left-most cell $c_\LLLL$\;
\While{the current cell is not the right-most cell $c_\RRRR$}
{
Move to the right cell\;
\While{the current cell is labeled $b-1$}
{
Write $0$ to the current cell\;
Move to the right cell\;
}
\If{the current cell is not the right-most cell $c_\RRRR$}
{
Let $x$ be the label of the current cell\;
Write $(x+1 \mod b)$ to the current cell\;
Move left repeatedly, until it reaches the left-most cell $c_\LLLL$\;}
}
 \caption{An algorithm that takes exponentially many steps.}\label{alg:simple}
\end{algorithm}

Although a time-constructable function $T(s)$ is defined only on positive integers, we extend its domain to all real numbers by setting $T(s) = T(\ceil{s})$ for each $s>0$ and $T(s) = T(1)$ for $s < 0$. In \cref{lem-time-hierarchy}, we relate time-constructable functions for Turing machines and complexity of $\LCL$ problems in the $\LOCAL$ model. The proof of \cref{lem-time-hierarchy} is deferred.

\begin{lemma}\label{lem-time-hierarchy}
Suppose $T(s)$ is a good time-constructable function with respect to linear-space-bounded Turing machines. 
\begin{itemize}
    \item If $T(s) = O(2^s / s)$, then there is an $\LCL$ problem that can be solved in $O\left(T\left( {\log n} \right)\right)$ rounds on general graphs and requires $\Omega\left(T\left( {\log n} \right)\right)$ rounds to solve on bounded-degree planar graphs.
    \item If $T(s) = \Omega(2^s)$, then there is an $\LCL$ problem that can be solved in $O\left(T\left({\log \log n} \right)\right)$ rounds on general graphs and requires $\Omega\left(T\left( {\log \log n} \right)\right)$ rounds to solve on bounded-degree planar graphs.
\end{itemize}
\end{lemma}

Combining \cref{lem-time-hierarchy,lem-time-hierarchy-aux}, we obtain the main result of the section.

\thmdense*
\begin{proof}
For any rational number $c > 0$, we pick a good time-constructable function $T(s) = \Theta(2^{cs})$  with respect to linear-space-bounded Turing machines, guaranteed by \cref{lem-time-hierarchy-aux}. 

We first consider the complexity classes  $\Theta(n^{c})$.
For $c < 1$, we have $T(s) = O(2^s / s)$, so \cref{lem-time-hierarchy} shows the existence of an $\LCL$ problem with complexity $\Theta\left(T\left( {\log n} \right)\right) = \Theta(n^{c})$, for all rational numbers $c$ such that $0 < c < 1$.  Here we use the property that $\GG$ includes all planar graphs, guaranteed by \cref{thm-excluding-grid}, so the lower bounds in \cref{lem-time-hierarchy} are applicable. For the extremal cases of $c = 0$ and $c = 1$, the existence of  $\LCL$ problems with complexity $\Theta(n^{c})$ is trivial.

For the complexity classes  $\Theta(\log^{c} n)$, as long as  $c \geq 1$, $T(s) = \Omega(2^s)$, so  \cref{lem-time-hierarchy} shows the existence of an $\LCL$ problem with complexity $\Theta\left(T\left( {\log \log n} \right)\right) = \Theta(\log^{c} n)$, for all rational numbers $c \geq 1$. 
\end{proof}

 The proof of \cref{lem-time-hierarchy} is given in \cref{sect-transform,sect:part1,sect:part2,sect:part3}.

\subsection{Transforming Turing Machines into Distributed Problems}\label{sect-transform}

Given a Turing machine $\MM = (Q, \Gamma, \{\bot_\LLLL, \bot_\RRRR\}, F,  q_{0},  \delta)$ that takes $T(s)$ steps to accept the input string $0^s$, in order to prove \cref{lem-time-hierarchy}, we need to construct  $\LCL$ problems realizing the round complexity bounds specified in \cref{lem-time-hierarchy}. 

We will consider $\LCL$ problems with radius $r > 1$ and with input labels. We allow edge labels and edge orientations. Having a radius greater than one seems necessary in our construction as we need to detect short-length cycles. The use of input labels is not necessary, as input labels can be simulated by bounded-degree subtrees attached to nodes. Similarly, the use of edge labels and edge orientations can be simulated by node labels on degree-$2$ nodes if we subdivide each edge $e=\{u,v\}$ into a path $(u,x,y,v)$. See~\cite{BalliuBCORS19,chang:LIPIcs:2020:13096,ChangP17} for details. For convenience, we still consider $\LCL$ problems with input labels and allow edge labels and edge orientation, as they simplify the presentation. 

To establish the dense region  $[\Theta(\log n), \Theta(n)]$, 
We will partition the node set $V$ of the underlying network $G=(V,E)$  into three parts $V = V_P \cup V_E \cup V_M$ according to input node labels, see \cref{fig:parts} for an illustration.

\begin{figure}
	\centering
	\includegraphics[width=\textwidth]{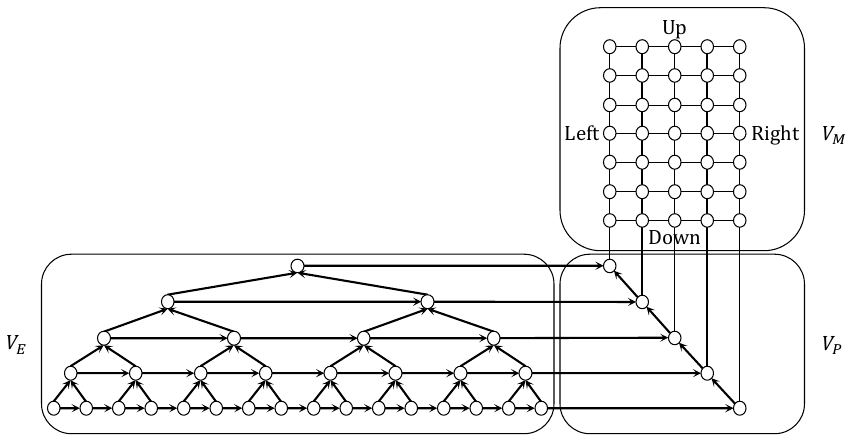}
	\caption{The three parts $V = V_P \cup V_E \cup V_M$.} 
	\label{fig:parts}
\end{figure}

We briefly explain the roles of these three parts. The node set $V_P$ forms a path that is seen as the input string for our Turing machine simulation. The node set $V_M$ forms a grid such that one of its four sides is attached to the path of $V_P$. The purpose of $V_M$ is to carry out the Turing machine simulation.
The node set $V_E$ has the structure of a complete binary tree such that the left-most root-to-leaf path is attached to the path of $V_P$. To use locally checkable rules to force $V_E$ to have the desired tree structure, the nodes in each layer will be connected into a path.
Note that the same kind of tree structure has been used in \cite{BBOS20paddedLCL} for defining a different $\LCL$ problem.

The purpose of $V_E$ is to increase the number of nodes in the graph to be exponential in the size of the input string of the Turing machine simulation.
By adding another complete binary tree, it is possible to further increase the number of nodes in the graph to be doubly exponential in the size of the input string, so we may handle the two cases of \cref{lem-time-hierarchy}.

Of course, the underlying network is still allowed to be any graph that does not follow the structure described above. As we will later see, all the rules specifying the desired graph structure are locally checkable, each node can check in $O(1)$ rounds whether the rules are met locally.

In particular, the rules applied to the nodes in $V_P$ are simple. Each $v \in V_P$ is required to have degree $1$ or $2$. Each edge $e=\{u,v\}$ whose two ends are in $V_P$ is required to be consistently oriented so that the nodes in $V_P$ form directed paths.

Intuitively, if  $V_P$ consists of only one path and all the rules are obeyed, then the string length will be $\Theta(\log n)$ or $\Theta(\log \log n)$, depending on which case of \cref{lem-time-hierarchy} we consider. The rules for the output labels will be designed in such a way that asks for simulation of a given Turing machine $\MM$ on the grid structure over  $V_M$, so the round complexities specified in  \cref{lem-time-hierarchy} can be achieved. 

As discussed earlier, the rules for the network topology and input label assignment might not be obeyed everywhere. To deal with these potential errors, the rules for the output labels will be designed in such a way that errors only make the problem easier to solve. For example, if the local neighborhood of a node  $v \in V_M$ does not form a grid structure, then $v$ will be exempted from doing the Turing machine simulation.

\subsection{Exponentiation} \label{sect:part1}
We describe the locally checkable rules applied to the nodes in $V_E$. The design objective of the rules is to ensure that the size of the $V_E$-component attached to a connected component $P$ of $V_P$ has size exponential in the number of nodes in $P$. Later we will discuss how the size can be made doubly exponential by modifying the rules.

\paragraph{Graph structure.}
We begin with a description of the graph structure that we want to realize with our locally checkable rules.
\begin{itemize}
    \item The nodes in $V_E$ are partitioned into the following three types based on their input labels: {\sf Top}, {\sf Middle}, and {\sf Bottom}.  
    \item The edges with at least one end in $V_E$ are oriented, and they are partitioned into two types: {\sf Tree} and {\sf Path}. 
\end{itemize}

We write $P = v_1 \leftarrow v_2 \leftarrow \cdots \leftarrow v_s$ to denote a connected component of $V_P$. 
The {\sf Tree} edges in $V_E$ induce a complete binary tree $T$ of exactly $s$ layers. These edges are oriented toward the direction of the root. Let $L_i \subseteq V_E$ be the set of nodes in the $i$th layer. Note that $|L_i| = 2^{i-1}$. The nodes in $L_i$ are connected into a directed path $u_{i,1} \leftarrow u_{i,2} \leftarrow \cdots \leftarrow u_{i, 2^{i-1}}$ using {\sf Path} edges. For each $1 \leq i < s$, $u_{i,j}$ is the parent of $u_{i+1, 2j-1}$ and $u_{i+1, 2j}$ in the complete binary tree $T$. The unique node in $L_1$ is a {\sf Top} node, the nodes in $L_2, L_3, \ldots L_{s-1}$ are {\sf Middle} nodes, and the nodes in $L_s$ are {\sf Bottom} nodes. For each $1 \leq i \leq s$, there is an edge $v_i \leftarrow u_{i,1}$ connecting $V_P$ and $V_E$. From the above description, the $V_E$-component attached to $P$ has exactly $2^{s} - 1$ nodes.

\paragraph{Basic rules.} We describe the locally checkable rules realizing the above graph structure. We start with the basic requirements about the {\sf Tree} and {\sf Path}  edges.  

\begin{itemize}
\item  All {\sf Top} and  {\sf Middle} nodes have exactly two incoming  {\sf Tree} edges. {\sf Bottom} nodes do not have incoming {\sf Tree} edges.  All {\sf Middle}  and  {\sf Bottom} nodes have exactly one outgoing {\sf Tree} edge.   {\sf Top} nodes do not have outgoing {\sf Tree} edges. From now on, we say that $u$ is a child of $v$ if there is a {\sf Tree} edge $v \leftarrow u$.

\item {\sf Top}  nodes do not have incident  {\sf Path} edges.
{\sf Middle} and  {\sf Bottom} nodes have one or two incident  {\sf Path} edges. In case they have two incident  {\sf Path} edges, one of them is incoming and the other one is outgoing.
\end{itemize}

\paragraph{Rules for the tree structure.}
To ensure that all the required  {\sf Tree} and {\sf Path}  edges exist, we require that 
each connected component $u_1 \leftarrow u_2 \leftarrow \cdots \leftarrow u_k$ induced by {\sf Path}  edges is properly 2-colored by $\{1,2\}$, encoded in input labels. Furthermore, we have the following requirements about the colors of the endpoints of the path. 
If a node  $v$ has one incoming  {\sf Path}  edge and zero outgoing  {\sf Path}  edge, then  $v$ is required to be colored by $1$.
If a node  $v$ has one outgoing  {\sf Path}  edge and zero incoming  {\sf Path}  edge, then  $v$ is required to be colored by $2$.
In case $v$ has no incident {\sf Path}  edges, $v$ is colored by $1$.
Using the proper 2-coloring, we make the following requirements about the local structure of the {\sf Tree} edges.

\begin{itemize}
    \item Let $u \leftarrow v$ be any {\sf Path}  edge where $\colors(u) = 1$ and $\colors(v) = 2$. It is required that there exists a node $w \in V_E$ whose two children are $u$ and $v$.
    
    Conversely, let $w  \in V_E$ be any node whose two children are $u$ and $v$. It is required that $u$ and $v$ are connected by a {\sf Path}  edge.
    
    \item Let $u \leftarrow v$ be any {\sf Path}  edge where $\colors(u) = 2$ and $\colors(v) = 1$. It is required that there exist two nodes  $w \in V_E$ and $x \in V_E$ with two incident  {\sf Tree} edges $w \leftarrow u$ and $x \leftarrow v$ and one {\sf Path}  edge $w \leftarrow x$. 
    
    Conversely, consider any two nodes  $w \in V_E$ and $x \in V_E$ connected by a {\sf Path}  edge $w \leftarrow x$ such that both $w$ and $x$ have two children. Let $u$ be the child of $w$ with $\colors(u) = 2$.  Let $v$ be the child of $x$ with $\colors(v) = 1$. It is required that $u$ and $v$ are connected by a {\sf Path}  edge in the direction $u \leftarrow v$.
    
    \item  Consider any $w \in V_E$ with two children $u$ and $v$ where $\colors(u) = 1$ and $\colors(v) = 2$. We have the following two requirements.
    \begin{itemize}
        \item $w$ has no outgoing {\sf Path}  edge if and only if $u$ has no  outgoing {\sf Path}  edge.
        \item $w$ has no incoming {\sf Path}  edge if and only if $v$ has no  incoming {\sf Path}  edge.
    \end{itemize}
\end{itemize}

Starting with any {\sf Top} node $v$ and applying the above rules inductively, it is ensured that there is a complete binary tree with the required structure where  $v$ is the root.

\paragraph{Rules for connecting $V_P$ and $V_E$.}
We describe the rules for the edges connecting  $V_P$ and $V_E$. Any edge connecting  $u \in V_P$ and $v \in V_E$ must be oriented in the direction  $u \leftarrow v$, and $V$ is required to have no outgoing {\sf Path} edges. Furthermore, we have the following requirements.
\begin{itemize}
    \item Each $v \in V_E$ with zero outgoing {\sf Path} edge has exactly one incident edge  $u \leftarrow v$ with $u \in V_P$. Conversely, each $u \in V_P$ has exactly one incident edge  $u \leftarrow v$ such that $v \in V_E$ has no outgoing {\sf Path} edges.
    \item Any edge $u \leftarrow v$ connecting  $u \in V_P$ and $v \in V_E$ belongs to one of the following types.
    \begin{itemize}
        \item $v$ has no outgoing {\sf Tree} edge ($v$ is {\sf Top}) and $u$ has no out-neighbor in $V_P$.
        \item $v$ has no incoming {\sf Tree} edge ($v$ is {\sf Bottom}) and $u$ has no in-neighbor in $V_P$.   
        \item $v$ has both incoming and outgoing {\sf Tree} edges ($v$ is {\sf Middle}) and $u$ has a in-neighbor and a out-neighbor in $V_P$.   
    \end{itemize}
    \item For each edge $u \leftarrow v$ connecting two nodes in $V_P$, there exists a {\sf Tree} edge $w \leftarrow x$ connecting two nodes in $V_E$ such that the two edges $u \leftarrow w$ and $v \leftarrow x$ exist.
\end{itemize}

\paragraph{Error pointers.}  Consider any connected component $P = v_1 \leftarrow v_2 \leftarrow \cdots \leftarrow v_s$ of $V_P$. If all of the locally checkable rules above are met, then the number of nodes in the $V_E$-component attached to $P$ is exactly $2^{s}-1$. Furthermore, due to the structure of the complete binary tree, the subgraph induced by $P$ and the  $V_E$-component attached to $P$ has diameter $O(s)$.  

It is possible that the underlying network does not satisfy all of the locally checkable rules. In that case, the first node $v_1$ of $P$ is able to detect an error in $O(s)$ rounds. Specifically, the errors can be reported using a special output edge label {\sf Error}. This label is allowed to appear in any directed edge $u \leftarrow v$, as long as one of the following requirements is met.
\begin{itemize}
    \item There is another edge $v \leftarrow w$ also labeled {\sf Error}.
    \item the $O(1)$-radius neighborhood of $v$ does not satisfy one of the above  rules.
\end{itemize}

If there is an error, then in $O(s)$ rounds we can find a string of error pointers $v_1 \leftarrow \cdots \leftarrow u$ from some node $u \in V_E$ where a rule is violated to the node $v_1$ of the path $P$. 

More specifically, due to the exponential growth of a complete binary tree, if we write $d$ to be the shortest length among all strings of error pointers, then the number of rounds it takes for $v_1$ to find any string of error pointers is $O(d)$. Because the graph structure within distance $d-1$ to $v_1$ is valid, the number of nodes seen by $v_1$ is at least $\Omega(2^d)$, due to the exponential growth of a complete binary tree, so the round complexity of finding such an error is always $O(\log n)$, where $n$ is the number of nodes in the network.

\paragraph{Double exponentiation.} Our rules ensure that when there is no error, the size of the $V_E$-component attached to $P = v_1 \leftarrow v_2 \leftarrow \cdots \leftarrow v_s$ is $2^{s}-1 = \Theta(2^s)$. It is straightforward to modify the rules to increase this number to $\Theta\left(2^{2^s}\right)$. The idea is to apply the same rules again to the directed paths induced by the {\sf Path} edges over the {\sf Bottom} nodes. Starting from the $s$-node path $P = v_1 \leftarrow v_2 \leftarrow \cdots \leftarrow v_s$ corresponding to a component of $V_P$, the $k = 2^{s-1}$ {\sf Bottom} nodes in the $V_E$-component attached to $P$ induces a $k$-node directed path $P'$ using {\sf Path} edges. If we apply the same rules again to these nodes, then we can further attach  $2^{k}-1 = \Theta\left(2^{2^s}\right)$ nodes to $P'$. This ensures that the total number of nodes in the $V_E$-component attached to $P = v_1 \leftarrow v_2 \leftarrow \cdots \leftarrow v_s$ is $\Theta\left(2^{2^s}\right)$, in case all rules are met.

\subsection{Turing Machine Simulation} \label{sect:part2}

We describe the locally checkable rules applied to the nodes in $V_M$. The design objective of the rules is to ensure that the nodes in the $V_M$-component attached to any connected component $P = v_1 \leftarrow v_2 \leftarrow \cdots \leftarrow v_s$ of $V_P$ form a grid structure on which we are required to simulate a given linear-space-bounded Turing machine $\MM$. If there is an error pointer pointed $v_1$, no simulation is needed. Otherwise, it is required that the output labels in $V_M$ constitute a correct simulation of $\MM$ on the input string $S = 0^{s-2}$.

\paragraph{Grid structure.} There are four directions: \upp, \downn, \leftt, \rightt, as we aim to build an oriented grid. Each edge $e = \{u,v\}$ connecting two nodes in $V_M$ is either in the \upp\ -- \downn\ direction or in the \leftt\ -- \rightt\ direction. If $u$ is the \leftt-neighbor of $v$, then $v$ is the \rightt-neighbor of $v$, and vice versa. If $u$ is the \upp-neighbor of $v$, then $v$ is the \downn-neighbor of $v$, and vice versa.

For each node $v$ in $V_M$, we write $d(v) = (d_\UUUU(v), d_\DDDD(v), d_\LLLL(v), d_\RRRR(v))$, where $d_\UUUU(v)$ is the number  \upp-neighbors of $v$, $d_\DDDD(v)$ is the number of \downn-neighbors of $v$, $d_\LLLL(v)$ is the number of \leftt-neighbors of $v$,  and  $d_\RRRR(v)$ is the number of \rightt-neighbors of $v$. It is required that $v$ belongs to one of the following cases.
\begin{itemize}
    \item $v$ is an inner node: $d(v) = (1,1,1,1)$.
    \item $v$ is on the \upp-border: $d(v) = (0,1,1,1)$.
    \item $v$ is on the \downn-border: $d(v) = (1,0,1,1)$.
    \item $v$ is on the \leftt-border: $d(v) = (1,1,0,1)$.
    \item $v$ is on the \rightt-border: $d(v) = (1,1,1,0)$.
    \item $v$ is on the (\upp, \leftt)-corner: $d(v) = (0,1,0,1)$.
    \item $v$ is on the (\upp, \rightt)-corner: $d(v) = (0,1,1,0)$.
    \item $v$ is on the (\downn, \leftt)-corner: $d(v) = (1,0,0,1)$.
    \item $v$ is on the (\downn, \rightt)-corner: $d(v) = (1,0,1,0)$.
\end{itemize}

For each node $u \in V_M$ with $d_\UUUU(u)=1$ and $d_\LLLL(u)=1$, it is required that there are three nodes $v$, $w$, and $x$ such that $v$ is the \upp-neighbor of $u$, $w$ is the \leftt-neighbor of $v$, $x$ is the \downn-neighbor of $w$, and $u$ is the \rightt-neighbor of $x$. We also have an analogous requirement for the rectangular faces in the three remaining directions, the details are omitted. These requirements ensure that the nodes in $V_M$ form disjoint unions of rectangular grids, if all the rules are obeyed everywhere.

\paragraph{Connections between $V_M$ and $V_P$.} We design rules to ensure that  any connected component $P = v_1 \leftarrow v_2 \leftarrow \cdots \leftarrow v_s$ of $V_P$ is glued to a grid of $V_M$ along its  \downn-border.  The rules are as follows.
\begin{itemize}
    \item Each $v \in V_P$ with zero out-neighbor in $V_P$ is required to be adjacent to a node in $V_M$ that is on the (\downn, \leftt)-corner. Conversely, each node in $V_M$ that is on the (\downn, \leftt)-corner is required to be adjacent to a node in $V_P$ with zero out-neighbor in $V_P$.
    \item Each $v \in V_P$ with zero in-neighbor in $V_P$ is required to be adjacent to a node in $V_M$ that is on the (\downn, \rightt)-corner. Conversely, each node in $V_M$ that is on the (\downn, \rightt)-corner is required to be adjacent to a node in $V_P$ with zero in-neighbor in $V_P$.
    \item Each $v \in V_P$ with one in-neighbor and one out-neighbor in $V_P$ is required to be adjacent to a node in $V_M$ that is on the \downn-border. Conversely, each node in $V_M$ that is on the \downn-border is required to be adjacent to a node in $V_P$ with one in-neighbor and one out-neighbor in $V_P$.
    \item For each edge $u \leftarrow v$ connecting two nodes in $V_P$, there exist two nodes $w$ and $x$ in $V_M$ such that $w$ is the \leftt-neighbor of $x$ and the two edges $\{u, w\}$ and $\{v, x\}$ exist.
\end{itemize}

Intuitively, if we write $P = v_1 \leftarrow v_2 \leftarrow \cdots \leftarrow v_s$ to denote a component of $V_P$ and let $u_i$ be the unique neighbor of $v_i$ in $V_M$, then $(u_1, u_2, \ldots, u_s)$ is the path in $V_M$ corresponding to the \downn-border of a grid: $u_1$ is on the (\downn, \leftt)-corner, $u_2, \ldots, u_{s-1}$ are on the \downn-border, and $u_s$ is on the (\downn, \rightt)-corner.

\paragraph{Error pointers.}  Similar to the errors in $V_E$, we allow the algorithm to use error pointers to let $v_1$ of $P = v_1 \leftarrow v_2 \leftarrow \cdots \leftarrow v_s$ learn errors in $P$ regarding its connections to $V_E$
or $V_M$ as well as errors in the nodes in  $V_M$ that are adjacent to $P$. 
To report such an error to $v_1$,  the error pointers will only go through the directed edges in $P$. Since the Turing machine simulation requires $s \geq 3$, the case of $s < 3$ will also be considered an error which can be reported by error pointers.
We cannot afford to report all errors in $V_M$ to $v_1$ via error pointers because the grid could be very large, and they will be handled differently.

\paragraph{Turing machine simulation.} Informally, there will be two possible ways to solve our $\LCL$ problem. The first option is to let $v_1$ of $P = v_1 \leftarrow v_2 \leftarrow \cdots \leftarrow v_s$ discover an error through a sequence of error pointers that ends at $v_1$. Then all nodes in the $V_M$-component attached to $P$ are allowed to output a special symbol $\ast$. The second option is to do a simulation of the given Turing machine $\MM$, where the input string is $0^{s-2}$. 

Intuitively, the \leftt\ -- \rightt\ dimension of the grid corresponds to the direction of the tape, and the \upp\ -- \downn\ dimension of the grid corresponds to the flow of time in the execution of the Turing machine.
Specifically, we write $u_{i,j}$ to denote each element in the grid such that $u_{1,1}$ corresponds to the node on the (\downn, \leftt)-corner and $u_{1,s}$ corresponds to the node on the (\downn, \rightt)-corner.
Then the simulation will start with labeling $u_{1,1}$ by $\bot_\LLLL$, labeling $u_{1,2}, \ldots, u_{1,s-1}$ by $0$, and labeling $u_{1,s}$ by  $\bot_\RRRR$. The head of the machine is initially over the left-most cell, which is $u_{1,1}$. We assume that the output labels of nodes in $V_M$ allow us to store an indicator as to whether the head of the machine is current at this cell and to store the current state of the machine if the head of the machine is at this cell.
The configuration of the machine after $k$ steps will be recorded by output labels in $u_{k+1, 1}, u_{k+1, 2}, \ldots, u_{k+1, s}$.

If there is no error in the grid, then the correctness of the output label at $u_{i,j}$ depends only on the output labels in $u_{i-1, j-1}$, $u_{i-1, j}$, and $u_{i-1, j+1}$, so the rules for the output labels are also locally checkable. 

\paragraph{Rules for the first layer.} We formally describe the locally checkable rules for the output labels at the nodes that are on the (\downn, \leftt)-corner, the \downn-border, and the (\downn, \rightt)-corner of the grid.
\begin{itemize}
    \item Let $v$ be on the (\downn, \leftt)-corner. 
    If the unique neighbor $u \in V_P$ of $v$ has an error pointer pointed to $u$, then $v$ is required to have $\ast$ as its output label.
    Otherwise, the output label of $v$ must indicate that it represents a cell whose label is $\bot_\LLLL$ and it is the current position of the head of the Turing machine.
    \item Let $v$ be on the \downn-border. If the \leftt-neighbor of $v$ has $\ast$ as its output label, then $v$ is also required to have $\ast$ as its output label, Otherwise, $v$ is required to have $0$ as its output label.
    \item Let $v$ be on the (\downn, \leftt)-corner. If the \leftt-neighbor of $v$ has $\ast$ as its output label, then $v$ is also required to have $\ast$ as its output label, Otherwise, $v$ is required to have $\bot_\RRRR$ as its output label.
\end{itemize}

Assuming that there is no error in the grid structure, there are only two possibilities for a correct output labeling of the first layer of the grid. The first possibility is that the output labels of the first layer of the grid correspond to the initial configuration of the Turing machine given $0^{s-2}$ as the input string. This case occurs if there is no error pointer pointed to $v_1$ of $P = v_1 \leftarrow v_2 \leftarrow \cdots \leftarrow v_s$.
The second possibility is that all nodes in the first layer of the grid are labeled $\ast$ in their output labeling. 
This case occurs only if an error is reported.

\paragraph{The use of the special symbol.}
Finally, we present some extra rules about the special output label  $\ast$  and how they resolve the remaining issues.
 As discussed earlier, the graph topology and input labels in $V_M$ might not meet all of the rules, and we cannot afford to report all these errors to $v_1$. 
Our solution is to \emph{always} allow each node in $V_M$ not meeting the rules locally to use the special label $\ast$ as its output label. Furthermore, we add the following rules for the output labeling. Recall that the correctness of the output label of a node $v \in V_M$ depends on the output labels of at most three other nodes. If at least one of these nodes labels themselves by $\ast$, then we are also allowed to label $v$ by $\ast$.

Since we do not have any rules to bound the size of the \upp\ -- \downn\ dimension of the grid,  we cannot afford to let all nodes in the grid communicate with the nodes in $P$, so they will not be able to know the outcome of the simulation of the Turing machine. 
To deal with this issue, we add the following rule. For each output label corresponding to the Turing machine entering an accepting state,  we change the output label to $\ast$. 
This means that if a node $v \in V_M$ is sufficiently far away from  $P$, then it can safely label itself $\ast$ without knowing the result of the Turing machine simulation.
Specifically, suppose the Turing machine simulation enters an accepting state on the $i$th layer of the grid. Due to the rule of the special symbol $\ast$, all nodes on the $(i+s-1)$th layer and after this layer have to label themselves $\ast$.

\subsection{Algorithms and Lower Bounds}\label{sect:part3}
In this section, we design algorithms and prove their optimality for our $\LCL$ problem.

\begin{proof}[Proof of \cref{lem-time-hierarchy}]
We are given a good time-constructable function $T(s)$ for linear-space-bounded Turing machines. 
Let $\MM$ be a Turing machine that takes $T(s)$ steps to accept the string $0^s$, and consider the $\LCL$ problem $\PP$ constructed using $\MM$. 
We will first consider the case where the size of the $V_E$-component attached to a $V_P$-component $P = v_1 \leftarrow v_2 \leftarrow \cdots \leftarrow v_s$ is $2^{s}-1 = \Theta(2^s)$ if all rules are met. The case where the size is doubly exponential will be discussed later.

\paragraph{Algorithm.}
Each node $v$ first locally checks whether $v \in V_P$ and has no out-neighbor in $V_P$. If so, then $v$ examine the set of nodes  $u \in V_P \cup V_E$ that are reachable to $v$ via a directed path $v \leftarrow \cdots \leftarrow u$ to check if there is any error.
Suppose there is an error. Pick $d$ to be the smallest distance of a directed path among all directed paths $v \leftarrow \cdots \leftarrow u$ using nodes in $V_P \cup V_E$ such that $u$ has an error. Then $v$ can detect such an error in $d$ rounds of communication. In this case, $v$ asks all edges in the directed path $v \leftarrow \cdots \leftarrow u$ to be labeled {\sf Error}.
Due exponential growth of the complete binary tree in the graph structure we discussed in \cref{sect:part2}, the number of nodes that $v$ sees within $d$ rounds is $O(2^d)$. Therefore, the round complexity of this step is always $O(\log n)$.

In case no error is found, we have $v = v_1$ for some connected component $P = v_1 \leftarrow v_2 \leftarrow \cdots \leftarrow v_s$ of $V_P$. The number of nodes in the $V_E$-component attached to $P$ is $2^s - 1$, so we have $s = O(\log n)$. Then $v$ spends $s-1 = O(\log n)$ rounds to inform all nodes in $P$ that there is no error. In case $u \in V_P$ does not receive such a message within $O(\log n)$ rounds, $u$ also learns that there is an error.

Each node $v_i$ in $P$ also checks whether there is any error in its connection to $V_M$.
If $P$ contains less than three nodes, then that also counts as an error.
If $v_i$ finds an error, it asks all the edges in the path $v_1 \leftarrow v_2 \leftarrow \cdots v_i$ to be labeled {\sf Error}, and $v_i$ also inform all nodes in  $P$ that there is an error.

For each $u \in V_P$, there are two cases. If $u$ has learned that there is an error during a previous step of the algorithm, then $u$ tells its neighboring node in $V_M$ to label itself $\ast$. Otherwise, $u$ tells its neighboring node in $V_M$ to label itself  $\bot_{\LLLL}$, $0$, or $\bot_{\RRRR}$, depending on whether $u=v_1$, $u \in \{v_2, \ldots, v_{s-1}\}$, or $u = v_s$ in the  connected component $P = v_1 \leftarrow v_2 \leftarrow \cdots \leftarrow v_s$ of $V_P$ that $u$ belongs to.

Now we consider any node $u \in V_M$. It first checks its local neighborhood to see if there is any error. If so, then it labels itself $\ast$.
From now on, we assume that the local neighborhood of $u$ obeys all the rules.
If $u$ has a neighbor $v$ in $V_P$, then it waits for $O(\log n)$ rounds until $v$ tells $u$ the output label of $u$.
Otherwise,  $u$ waits until its \downn-neighbor $v$, the $\leftt$-neighbor of $v$, and the $\rightt$-neighbor of $v$ decide their output labels. Once these three nodes have decided their output labels, $u$  decides its output label as follows. 
If one of these three output labels is $\ast$, then $u$ also labels itself $\ast$. Otherwise, $u$ computes its output label according to the transition rules of the given Turing machine $\MM$.  
As discussed earlier, if the computed output label indicates that the machine enters an accepting state, then $u$ labels itself $^\ast$.
If $u$ has waited for $T(s-2) + O(\log n)$ rounds and it still cannot decide its output label, then it labels itself $\ast$.

The round complexity of the algorithm is $T(s-2) + O(\log n) = O\left(T\left( {\log n} \right)\right)$, as $T(s- 2) = O\left(T\left( {\log n} \right)\right)$ because $s \leq \log n$, as the $V_E$-component attached to $P$ already has $2^s - 1$ nodes if we do the Turing machine simulation. The term $O(\log n)$ is not dominant, as we have $T(\log n)  \geq \log n$.
Although the desired graph structure when all the rules are met is planar, the error detection part of the algorithm does not rely on any assumption about the graph structure, so the algorithm works on any graph.

\paragraph{Lower bound.}

We will show that as long as  $T(s) = O(2^s / s)$, the $\LCL$ problem $\PP$ that we construct needs $\Omega\left(T\left( {\log n} \right)\right)$ rounds to solve on bounded-degree planar graphs.

The lower bound graph is simply the case where all the rules are obeyed and the length of the grid in the \upp\ -- \downn\ dimension is $2 T(s-2)$, measured in the number of nodes. We assume $s \geq 3$. In this case, $|V_P| = s$, $|V_E| = 2^{s} - 1$, and $|V_M| = s \cdot 2T(s-2)$. As we assume that $T(s) = O(2^s / s)$ for all $s$, the total number of nodes in this graph is $O(2^s)$, so $s \geq \log n - O(1)$.

Because there is no error, no error pointer can be used as output labels. Therefore, the rules of our  $\LCL$ problem $\PP$  imply that we have to do a simulation of $\MM$ with the input string $0^{s-2}$.
We write $u_{i,j}$ to denote each element in the grid such that $u_{1,1}$ corresponds to the node on the (\downn, \leftt)-corner and $u_{1,s}$ corresponds to the node on the (\downn, \rightt)-corner.
Suppose in an execution of $\MM$ with the input string $0^{s-2}$, it enters an accepting state after $i=T(s)$ steps at the $j$th cell. 
Then in a correct output labeling of the graph, $u_{i+1,j}$ is labeled $\ast$, and $u_{i,j}$ is not labeled $\ast$. However, the $(T(s-2)-1)$-radius neighborhood of  $u_{i,j}$ and $u_{i+1,j}$ are exactly the same. Hence no randomized algorithm can solve $\PP$ on this graph successfully with probability greater than $1/2$ in $T(s-2)-1$ rounds. As $s \geq \log n - O(1)$, we obtain the $\Omega\left(T\left( {\log n} \right)\right)$ lower bound, using the fact that $T(s)$ is good.

Hence we conclude the proof of the first part of \cref{lem-time-hierarchy}. That is, given that $T(s) = O(2^s / s)$, there is an $\LCL$ problem that can be solved in $O\left(T\left( {\log n} \right)\right)$ rounds on general graphs and requires $\Omega\left(T\left( {\log n} \right)\right)$ rounds on bounded-degree planar graphs.

\paragraph{Double exponentiation.}  For the rest of the section, we consider the case where the size of the $V_E$-component attached to a $V_P$-component $P = v_1 \leftarrow v_2 \leftarrow \cdots \leftarrow v_s$ is doubly exponential $\Theta(2^{2^s})$ if all the rules are met. The algorithm is the same as the one above, which has the round complexity $T(s-2) + O(\log n)$.

 Note that $s \leq \log \log n + O(1)$, as the $V_E$-component attached to $P$ already has $2^{2^{s-1}} - 1$ nodes.
If we assume that $T(s) = \Omega(2^{s})$ for all $s$, then $T(s-2) = \Omega(\log n)$. Hence we may write $T(s-2) + O(\log n) = O\left(T\left( {\log \log n} \right)\right)$.

For the lower bound, we also consider the graph where all the rules are obeyed and the length of the grid in the \upp\ -- \downn\ dimension is $2 T(s-2)$. We have $|V_P| = s$, $|V_E| = 2^{2^{s-1}} - 1$, and $|V_M| = s \cdot 2T(s-2)$. The assumption that $T(s)$ is good implies that $T(s) = 2^{O(s)}$, so the number of nodes in the graph is $n = O\left(  2^{2^{s}}\right)$. That is, $|V_E|$ is the dominant term. This means that we have $s \geq \log \log n - O(1)$. Hence we have the lower bound $T(s-2) - 1 = \Omega\left(T\left( {\log \log n} \right)\right)$.

Hence we conclude the proof of the second part of \cref{lem-time-hierarchy}. That is, given that  $T(s) = \Omega(2^s)$,  there is an $\LCL$ problem that can be solved in $O\left(T\left({\log \log n} \right)\right)$ rounds on general graphs and requires $\Omega\left(T\left( {\log \log n} \right)\right)$ rounds to solve on bounded-degree planar graphs.
\end{proof}

\end{document}